\documentclass[aps,prx,superscriptaddress,nofootinbib,notitlepage,showpacs,floatfix,twocolumn]{revtex4-1}

\usepackage{graphicx,graphics,epsfig,subfigure,times,bm,bbm,amssymb,amsmath,amsfonts,mathrsfs}
\usepackage[scr=boondoxo,scrscaled=1.05]{mathalfa}
\usepackage[matrix,frame,arrow]{xypic}
\usepackage[pdfstartview=FitH]{hyperref}

\usepackage[pdftex]{color}

\newtheorem{theorem}{Theorem}

\newtheorem{definition}{Definition}
\newtheorem{example}{Example}

\newcommand{\beq}{\begin{equation}}
\newcommand{\eneq}{\end{equation}}
\newcommand{\beqnn}{\begin{equation*}}
\newcommand{\eneqnn}{\end{equation*}}
\newcommand{\beqy}{\begin{eqnarray}}
\newcommand{\eneqy}{\end{eqnarray}}
\newcommand{\beqynn}{\begin{eqnarray*}}
\newcommand{\eneqynn}{\end{eqnarray*}}
\newcommand{\half}{\mbox{$\textstyle \frac{1}{2}$}}

\newcommand{\ket}[1]{ | #1 \rangle  }
\newcommand{\bra}[1]{ \langle  #1 | }

\newcommand{\tr}[1]{\textrm{Tr}[ {#1} ]}

\newcommand{\T}{\textrm{T}}

\newcommand{\bes} {\begin{subequations}}
\newcommand{\ees} {\end{subequations}}
	\newcommand{\bea} {\begin{eqnarray}}
	\newcommand{\eea} {\end{eqnarray}}

\newenvironment{proof}[1][Proof]{\noindent\textbf{#1.} }{\ \rule{0.5em}{0.5em}}

\newcommand{\ignore}[1]{}

\usepackage{ulem}

\newcommand{\blk}{\color{black}}

\definecolor{maroon}{rgb}{0.7,0,0}

\definecolor{ngreen}{rgb}{0.3,0.7,0.3}

\definecolor{golden}{rgb}{0.8,0.6,0.1}

\newcommand{\R}{{\cal R}}

\begin{document}

\title{On the dynamics of initially correlated open quantum systems: theory and applications}

\author{Gerardo~A. Paz-Silva}

\affiliation{Centre for Quantum Computation and Communication Technology (Australian Research Council), \\ 
Centre for Quantum Dynamics, \blk Griffith University, Brisbane, Queensland 4111, Australia} 
 \author{Michael J. W. Hall}
 \affiliation{Centre for Quantum Computation and Communication Technology (Australian Research Council), \\ 
Centre for Quantum Dynamics, \blk Griffith University, Brisbane, Queensland 4111, Australia} 
\affiliation{Department of Theoretical Physics, Research School of Physics and Engineering,
	Australian National University, Canberra ACT 0200, Australia}

\author{Howard M. Wiseman}
\affiliation{Centre for Quantum Computation and Communication Technology (Australian Research Council), \\ 
Centre for Quantum Dynamics, \blk Griffith University, Brisbane, Queensland 4111, Australia} 

\begin{abstract}
We show that the dynamics of any open quantum system that is initially correlated with its environment can be described by a set of $d^2$ (or  less)  completely positive maps, where $d$ is the dimension of the system. Only one such map is required for the special case of no initial correlations. The same maps describe the dynamics of any system-environment state obtained from the initial state by a local operation on the system. \ignore{(and are related to the set of states to which the environment may be steered via such local operations).} The reduction of the system dynamics to a set of completely positive maps allows known numerical and analytic tools for uncorrelated initial states to be applied to the general case of initially correlated states, which we exemplify by solving the qubit dephasing model for such states, \blk and provides a natural approach to quantum Markovianity for this case. We show that this set of completely positive maps can be experimentally characterised using only local operations on the system, via a generalisation of noise spectroscopy protocols. As further \blk applications, we first consider the problem of retrodicting the dynamics of an open quantum system which is in an arbitrary state when it  becomes  accessible to the experimenter, and explore the conditions under which retrodiction  is possible. We also introduce a related one-sided or limited-access tomography protocol for determining an arbitrary bipartite state, evolving under a sufficiently rich Hamiltonian, via local operations and measurements on just one component. We simulate this  protocol for  a physical model of particular relevance to nitrogen-vacancy centres, and in particular show how to  reconstruct the density matrix of a set of three qubits, interacting via dipolar coupling and in the presence of local magnetic fields, by  measuring and controlling only one of them.
\end{abstract}

\maketitle


\section{Introduction} \label{introduction}

Understanding, accurately predicting, and controlling the (average or expected) behavior of quantum systems in realistic scenarios, i.e., in the presence of noise, is fundamental to the development of quantum-enhanced cutting-edge technologies such as quantum computing~\cite{nielsen} and quantum metrology~\cite{QSensing}. This is the domain of the theory of open quantum systems and quantum control.  

Mathematically, the theory of open quantum systems deals with the general scenario in which a quantum system, $S$,  interacts with (typically inaccessible) external degrees of freedom, dubbed the environment or bath,  $B$. The state of the system plus bath is described by a density matrix $\rho_{SB}(t)$ on a joint Hilbert space \blk $\mathcal{H}_S \otimes \mathcal{H}_B$, whose evolution  ruled by a Hamiltonian $H(t)$ via the unitary operation $U(t) = \mathcal{T} e^{-i \int_{0}^t H(s) ds}$, with $\rho_{SB}(t) = U(t) \rho_{SB}(0) U^\dagger(t)$. Under these conditions, the objective is to predict and eventually control the reduced dynamics of the system in the presence of the inaccessible bath. That is, one would like to determine    
\beq \label{reducedintro}
\rho_S(t) = {\rm Tr}_B [U(t) \rho_{SB} (0) U(t)^\dagger], 
\eneq
for any $t$. This is a highly non-trivial problem which can only be analytically solved in very special scenarios and/or under strong assumptions such as Gaussianity. Powerful analytical and numerical methods to solve this problem, generally approximately, have been devised over the years, including various flavors of master equation and path integral methods~\cite{OpenQuantumSystemsReview,OpenQuantumSystemsBook1,OpenQuantumSystemsBook2}. 

In deriving and applying such methods, two strong assumptions are typically made. The first is the so-called {\it factorisable initial state} condition, i.e., $\rho_{SB}(0) = \rho_S \otimes \rho_B$. The second assumption is one of {\it sufficient dynamical information}: noting that the bath is, in general, not fully accessible, sufficient knowledge about $\rho_B$ and $H(t)$ must be assumed to determine their effect on the system dynamics. In this paper we will show that,  surprisingly, \blk these two assumptions can both be dispensed with. This allows the evolution of {\it arbitrary} open systems to be characterised using completely positive maps and quantum sensing protocols, and opens the way  for applications such as \blk the tomography of bipartite systems via measurements on one side only. 

\subsection{Summary of  results}

The factorisable initial state assumption $\rho_{SB}(0) = \rho_S \otimes \rho_B$ is rather strong, but is useful as it implies that 
\beq
\label{CPTPfact}
\rho_{S}(t) = \phi_t(\rho_S), 
\eneq
where  $\phi_t(\cdot)$ is a completely positive trace preserving (CPTP) map. This assumption is ubiquitous in the theory of open quantum systems, and underpins the widespread use of CPTP maps in quantum information theory~\cite{nielsen} and  in definitions of quantum Markovianity~\cite{RHPreview,Breuerreview,LHWreview}. In contrast, \blk for initially correlated states it is not possible to describe the system evolution in this way, other than for an extremely limited class of initial states~\cite{pechukas,alicki,shaji,CP1,CP2,CP3,CP4,CP5,CP6,CP7,CP8,CP9,CP10,CP11,CP12}. What is more,  calculation methods  used to (even approximately) solve Eq.~\eqref{reducedintro} successfully, in the factorisable case, do not typically apply to correlated initial states \blk (although there are interesting exceptions in certain special cases~\cite{breuer07}). \blk	

Our first main result is to show that CPTP maps acting on states remain a useful tool  even for initially correlated \blk states, with the evolution of a $d$-dimensional  open quantum \blk system requiring only $d^2$ such maps at most.  Moreover, for a given initial state $\rho_{SB}$, the same set of maps describes the system evolution for {\it any} initial state obtained from $\rho_{SB}$ via a local operation on the system. This is a very large space of states  (typically on the order of $d^4$ dimensions), \blk which includes all factorisable states $\tau_S\otimes\rho_B$ (for arbitrary $\tau_S$ and \blk \blk $\rho_B:={\rm Tr}_S[\rho_{SB}]$), and in general many others. This  result \blk is based on a natural ``bath-positive'' decomposition of initially-correlated states,  and allows \blk existing calculation methods for solving Eq.~\eqref{reducedintro} in the factorisable case  to be \blk easily  extended to the arbitrary initial condition scenario. \blk We exemplify this by fully solving the problem of qubit dephasing~\cite{massimo, reina, OpenQuantumSystemsBook2} in the case of arbitrary initial correlations.   We stress that the key feature of our decomposition is that it applies to arbitrary initial states and is thus constructive and universal, which is in stark contrast with other decompositions that have been introduced in the literature (see for example Refs~\cite{deco1},~\cite{deco2}) to overcome the correlated initial state problem in particular scenarios. \blk We also demonstrate a direct link between the set of maps and quantum steering~\cite{steer}, and with the superchannel recently introduced by Modi that maps local system operations to the system state at later times~\cite{modi}. We further explore quantum Markovianity in the context  of initial correlations,  and define a notion of ``computational Markovianity'' in terms of the Markovian character of the set of the CPTP maps describing the dynamics.   We explore these results in Sections~\ref{theory} and \ref{Dynamics}.

Our second main result is to show that one can remove the {\it sufficient dynamical information} assumption. In the factorisable case,  it is known that \blk this assumption can be bypassed in the sense that, while $\rho_B$ and $H(t)$ cannot be directly measured in general, the necessary information for describing the system evolution \blk can be {\it indirectly} measured. Indeed, this is the motivation behind a recent push to develop so-called {\it noise spectroscopy protocols}~\cite{Spectro1,Spectro2,Spectro3,Spectro4,Spectro5,Spectro6,Spectro7,Spectro8,Spectro9,Spectro10,Spectro11}. These protocols are based on the observation that the dynamics of the system does not require explicit knowledge of $H(t)$ and $\rho_B$, but rather of the {\it correlations} present in the bath. In particular, writing the joint system bath Hamiltonian in the interaction picture with respect to the bath self Hamiltonian $H_B$ as $H(t) = \sum_{b} W_b \otimes B_{b}(t)$, with $\{W_b\}$  an operator basis for the system, the dynamics of the system  for a factorisable initial state $\rho_{SB}(0) = \rho_S \otimes \rho_B$ \blk depend only on the \blk bath {\it correlators}~\cite{OpenQuantumSystemsReview}
\beq
\label{Corre}
\langle B_{b_1} (t_1) \cdots B_{b_k} (t_k) \rangle  \equiv \langle \tr{ B_{b_1} (t_1) \cdots B_{b_k} (t_k) \rho_B} \rangle_c,
\eneq
where $\langle \cdot \rangle_c$ denotes the average over realisations of any (classical) stochastic  processes for $B_b(t)$. Importantly,  if detailed information about such correlators is available, it is possible to design control sequences (via optimal control techniques for example) capable of executing a desired system operation with high fidelity~\cite{FiltersApp}. 

Noise spectroscopy protocols exploit the ability to measure the response of the quantum system to different control sequences in the presence of the bath, in order to obtain the Fourier transforms, $\langle \tilde B_{b_1} (\omega_1) \cdots \tilde B_{b_k} (\omega_k) \rangle$, of the aforementioned bath correlators. To date, detailed protocols have only been described for certain noise models, i.e., for $H(t)$ and $\rho_B$ satisfying specific conditions. Nevertheless the general methodology behind them allows, \blk in principle, a protocol for general noise models to be designed. More broadly, protocols that exploit the ability to measure the response of a quantum system to its environment, be it classical or quantum, are the essence of {\it quantum sensing}~\cite{QSensing} and they range from the simpler phase or parameter estimation protocols (for a constant and classical $B(t)$) to the more ambitious noise spectroscopy protocols outlined above. Here, we will show  how quantum sensing protocols can be seamlessly extended \blk to the scenario where system and bath are initially correlated, thus dispensing with both the factorisable initial state and dynamical information assumptions. \blk These results are contained in Section~\ref{QSense}.

In Secs.~\ref{retrosec} and~\ref{LAT} we give two applications of  such extended quantum sensing protocols.  First, in Sec.~\ref{retrosec},  we show practical retrodiction of the system state at earlier times is possible, under mild assumptions on the system-bath Hamiltonian $H(t)$. \blk
Finally, as detailed in Section~\ref{LAT}, we \blk develop a one-sided tomography protocol. Concretely, we show how, when the bath is finite dimensional and $H(t)$ is known and sufficiently non-trivial, it is possible to do tomography on the joint system-bath state despite having only control and measurement capabilities on the system. We describe this in detail for  three interacting \blk qubits  with dipole-dipole couplings in a magnetic field --- of particular relevance to nitrogen-vacancy (NV) centers~\cite{NVcenter} and  nuclear magnetic resonance (NMR)~\cite{NMR}. \blk

\subsection{Practical motivation}

Having  briefly outlined the scope of our results, we now comment on the practical need to study the problem considered here. This is complementary to \blk the standard motivations given in the literature~\cite{CP1, CP2, modi}, which  revolve around the fundamental question in the theory of open quantum systems: what is $\rho_S(t)$ given an arbitrary initial condition $\rho_{SB}$?  It is also additional to  the recent practical \blk uses \blk of correlated initial states for \blk consistent calculations of condensed-phase reaction rates~\cite{tannor}  and \blk for engineering arbitrary phase decoherence dynamics~\cite{piiloexp}. \blk

As hinted in the summary of results, characterizing the dynamics of a quantum system that is initially correlated with its environment \blk is crucial on our road to the development of quantum technologies. An idealised text-book quantum computer can be described by an initial preparation stage, followed by a unitary gate $U,$ and finalised by a measurement stage that extracts the result of the computation~\cite{nielsen}. With the realisation that operations are never ideal and that the coupling of quantum systems to a bath induces noise, the unitary operation stage is then usually replaced by a quantum channel \blk $\Lambda_U$. Standard methods to characterise the error induced by the presence of the bath crucially rely on, among other conditions, the assumption \blk that $\Lambda_U$ is a CPTP map (or a sequence of them). Randomised benchmarking~\cite{RB}, for example, heavily uses the assumption that a sequence of noisy unitaries can be modeled by a sequence of CPTP maps. This clearly constraints the possible set of noise models that can be analyzed by the tool as, for example, if noise is introduced by a quantum bath such assumption cannot hold as after even a single unitary the system and bath become entangled. 
 Additionally, error correction and error suppression techniques~\cite{DD1,DD2,DD3, DDBook}, developed to ensure that $\Lambda_U$ is as close as possible  to the ideal unitary operation (in an appropriate sense, e.g., the diamond norm), usually implicitly assume that the system and bath are initially in a factorisable state. In particular, noise spectroscopy protocols and the associated optimal control techniques, which have seen a big push in recent years~\cite{Spectro1,Spectro2,Spectro3,Spectro4,Spectro5,Spectro6,Spectro7,Spectro8,Spectro9,Spectro10,Spectro11}, have been developed in the context of this assumption. \blk However,  this CPTP \blk assumption does not correspond to \blk the most general scenario~\cite{OpenQuantumSystemsReview} and, at the very least, has to be verified. In summary, pushing quantum system \& noise characterisation protocols and methods to predict/retrodict the evolution of a quantum system beyond the single CPTP scenario and into in the more realistic general non-factorisable state setting is the main practical \blk  motivation of this work.  

Now, an argument might be made that any ``good'' preparation procedure should initialise the system-bath in a factorisable state, {\blk potentially rendering the above critique trivial. However, such an argument does not overcome the problem when one is trying to understand a sequence of unitary operations in the presence of the bath. After even a single unitary, system and bath generally become entangled, and thus  at any given time $t>0$ the factorisable initial condition cannot hold in general.}
	
\blk Further,  even at the $t=0$ preparation stage,  there is a subtle but important problem with  such an argument.  To see this consider perhaps the simplest preparation procedure that comes to mind, that ``naturally'' and deterministically outputs a desired factorisable state of the form $\ket{\psi_0}\bra{\psi_0} \otimes \rho_B$ for an $n$-qubit system, \blk with, e.g., $\ket{\psi_0} =\ket{+}^{\otimes n}$. Starting from a \blk typically-correlated \blk initial state $\rho_{SB}$,  one applies a projective measurement with $2^n$ outcomes on the system and then, based on the outcome, applies a unitary operation that rotates each qubit \blk to the desired state. For the $i$-th outcome we will have 
$$ \rho_{SB}^{(i)} = \ket{\psi_0}\bra{\psi_0} \otimes \frac{{\rm tr}_S [(\Pi_i \otimes I_B ) \rho_{SB}]}{{{\rm tr} [(\Pi_i \otimes I_B ) \rho_{SB}]}} = \ket{\psi_0}\bra{\psi_0} \otimes \rho_B^{(i)},$$ 
where $\Pi_i$ is the projector on the state corresponding the $i$th outcome, and we have already included the effect of the rotation in the preparation procedure. The key feature, arising from the fact that the state before the preparation procedure is non-factorisable, is that the evolution of the system after the preparation procedure will depend on the measurement outcome, via $\rho_B^{(i)}$. \blk This is in stark contrast with the factorisable initial state scenario. Further, if we only concern ourselves with the output of the preparation procedure $\ket{\psi_0}\bra{\psi_0}$ (effectively throwing away information about the measurement outcome), and attempt to characterise the subsequent system evolution, we find that it effectively evolves under a CPTP map generated by the bath state $\rho_B = \sum_i {\rm tr} [(\Pi_i \otimes I_B ) \rho_{SB}] \rho_B^{(i)}$. This  is not, however, the correct CPTP map ruling the evolution of the state after each preparation, thus pointing to the need to characterise all possible bath states after the projective measurements or, more generally, to understand in more detail the structure of the initially correlated state. While using $\rho_B$ (and not the $\rho_B^{(i)}$)  to predict the evolution of $\ket{\psi_0}\bra{\psi_0}$ may not be a problem for understanding the gross behavior of a quantum system, it will be crippling when attempting to achieve high-fidelity gates via noise spectroscopy methods and related optimal control methods discussed earlier. 

{\blk  Thus, there are practical motivations to go beyond  the initially factorizable state assumption.}
The results we described above, and that we now proceed to explain, are key to \blk doing so, \blk and open the way to generalisations of well-established protocols, such as randomised benchmarking and noise spectroscopy, to the most general setting.

\section{Bath-positive decompositions}
\label{theory}

\subsection{General definition}


The key to the developments in this paper will be the ability to decompose an arbitrary initial density matrix as 
\begin{equation} \label{bplusgen}
\rho_{SB}(0) = \sum_\alpha w_\alpha Q_\alpha \otimes \rho_\alpha,
\end{equation}
where, crucially, each $\rho_\alpha$ is a valid density matrix of the bath and $\{Q_\alpha\}$ forms a (possibly overcomplete) basis for operators on $\mathcal{H}_S$.
The $Q_\alpha$ are not restricted to be positive or trace-orthogonal.  Note that in this form all information about the initial system state, \blk
\beq \label{rhos}
\rho_S:={\rm Tr}_B[\rho_{SB}(0)]=\sum_\alpha w_\alpha Q_\alpha ,
\eneq
is condensed into the coefficients $\{w_\alpha\}$, while information about  correlations also resides in the $\{ \rho_{a}\}$. It should be noted that this is not the only way of generating a decomposition for $\rho_{SB}$ such that in each term the bath component is a density operator but, as we will see, the fact that, in addition, the $\{Q_\alpha\}$ form a fixed operator basis  is crucial for our results. We dub this a {\it bath-positive} or B+ decomposition of the density matrix. To illustrate this, let us consider a finite dimensional example. 

\begin{example}  \label{Ex1} (Qubit plus bath) When the system  is \blk a qubit, one can  use the completeness of the Pauli sigma basis $\{\sigma_0, \sigma_x,\sigma_y,\sigma_z\}$ (with $\sigma_0=\mathbf{1}$) to write an arbitrary joint state \blk $\rho_{SB}$ as 
\begin{align*}
\rho_{SB} &=  \half \sum_{ \alpha \blk =0,x,y,z} \sigma_{ \alpha \blk} \otimes {\rm tr}_S[(\sigma_{ \alpha \blk}\otimes\mathbf{1}_B)\rho_{SB}] \blk \\
& =: \half \sum_{ \alpha \blk =0,x,y,z} \sigma_{ \alpha \blk} \otimes \eta_{ \alpha \blk} \\
&= \frac{\sigma_0 - \sum_{ \alpha \blk =x,y,z} \sigma_{ \alpha \blk }}{2} \otimes \eta_0 + \sum_{ \alpha \blk =x,y,z} \frac{\sigma_{ \alpha \blk }}{2} \otimes (\eta_0 + \eta_{ \alpha \blk })\\
&\equiv \sum_{\alpha=0,x,y,z} w_\alpha Q_\alpha \otimes \rho_\alpha, 
\end{align*}
with weights $w_\alpha$ and density operators $\rho_{ \alpha \blk}$ defined via $p_0\rho_0=\eta_0$, and $p_{\alpha}\rho_\alpha =\eta_0+\eta_{ \alpha \blk}$ for $\alpha=x,y,z$. Note via the second line that \blk $\eta_0 = {\rm tr}_S [\rho_{SB}]=\rho_B$, and  $\eta_0 + \eta_{ \alpha \blk } = {\rm tr}_S [ ((\sigma_0 + \sigma_{ \alpha \blk })\otimes \mathbf{1}_B) \rho_{SB}]$ for $\alpha=x,y,z$, and so  $\rho_{\alpha}$ is a positive operator as desired. Similar constructions can be crafted for higher dimensions, using generalised Pauli bases. 
\end{example}

A general construction of B+ decompositions is as follows. First, let  $\{ P_\alpha\}$ be any basis set of positive system operators. This basis set may be overcomplete, and is also called a  {\it frame}~\cite{frames,optTomogra}. For such $\{ P_\alpha\}$ one can always construct a {\it dual basis} or {\it dual frame}, $\{Q_\alpha\}$, such that  that any system operator $A$ acting on $\mathcal{H}_S$ can be decomposed as 
\beq
\label{decom}
A = \sum_{\alpha} \tr{ A Q_\alpha} P_\alpha =  \sum_{\alpha} \tr{ A P_\alpha} Q_\alpha.
\eneq
In particular, if  $\{G_j\}$ is an orthonormal basis set of Hermitian operators  on $\mathcal{H}_S$, with $\tr{G_j G_k}=\delta_{jk}$, then a suitable dual frame is specified by \blk (see Appendix~\ref{construct}) \blk
\beq \label{qmp}
Q_\alpha =\sum_\beta M_{\alpha\beta} P_\beta ,\qquad \textrm{M}=\textrm{T}\,(\textrm{T}^\top \textrm{T})^{-2}\, \textrm{T}^\top,
\eneq
where $\textrm{T}$ is the (typically non-square) matrix with coefficients $\textrm{T}_{\alpha j} := \tr{P_\alpha G_j}$.  
It is important to highlight that while we focus on finite dimensional systems in this paper,  the above construction also applies to infinite dimensional Hilbert spaces under a mild condition on $\{P_\alpha\}$ (see Appendix~\ref{construct}). 

From the above, it immediately  follows that any joint state $\rho_{SB}$ on $\mathcal{H}_S \otimes \mathcal{H}_B$ has a corresponding B+ decomposition 
\beq
\label{maindeco}
\rho_{SB} = \sum_\alpha Q_\alpha \otimes {\rm Tr}_S[(P_\alpha\otimes 1_B)\rho_{SB}] \equiv \sum_\alpha w_\alpha Q_\alpha \otimes \rho_\alpha,
\eneq
where the weights $w_\alpha$ and the bath density operators $\rho_\alpha$ are implicitly defined via
\begin{align} \label{rhoalpha}
w_\alpha  \rho_\alpha &= {\rm Tr}_S[(P_\alpha\otimes 1_B)\rho_{SB}]
\end{align} 
(with $\rho_\alpha$ arbitrary when the right hand side vanishes). Note that taking the trace of Eq.~(\ref{rhoalpha}) over the bath yields 
\beq
\label{weights}
w_\alpha={\rm Tr}_S[P_\alpha \rho_S].
\eneq 
For the special scenario of a factorisable state one has that $\rho_\alpha= \frac{\tr{(P_{\alpha} \otimes \bm{1})  \rho_{S} \otimes \rho_B}}{\tr{P_{\alpha}\rho_{S}}}=  \rho_B$ for all \blk $\alpha$, as expected.

\subsection{Canonical B+ decompositions}
\label{canonbplus}

Of particular interest are frames $\{ P_\alpha \}$ for which the basis elements $P_\alpha$ are linearly independent (i.e., with precisely $d^2$ basis elements for the case of a $d$-dimensional system Hilbert space), as in Example~\ref{Ex1} above.  Since the expansion of any operator in such a basis is unique, the dual frame is also unique, and taking $A= Q_\beta$ in Eq.~\eqref{decom} implies the biorthogonality property
\beq \label{biorthog}
\tr{ P_\alpha Q_{\beta}} = \delta_{\alpha\beta}.
\eneq
Linear independence further implies that the matrix $\T$ in Eq.~(\ref{qmp}) is invertible, yielding $\textrm{M}=(\T^\top\T)^{-1}$ for the matrix connecting the frame with its dual.

If, additionally, the basis elements $\{ P_\alpha \}$ form a positive operator valued measure (POVM) on $\mathcal{H}_S$, i.e., $\sum P_\alpha ={\bm 1}_S$, then from Eq.~(\ref{weights})  the weights $\{w_\alpha\}$ have a simple interpretation as the probability distribution corresponding to a measurement of $\{P_\alpha\}$ on the system, with 
\beq \label{normal}
\blk w_\alpha \geq 0,\qquad \blk \sum_\alpha w_\alpha = 1. 
\eneq
Further, \blk $\rho_\alpha$ corresponds to the conditional state of the bath for measurement outcome $\alpha$. \blk
Note also that, since the POVM elements form a basis set, the POVM is informationally complete, i.e., the statistics of the measurement are sufficient to reconstruct the initial \blk system density operator via Eq.~(\ref{rhos}).

Thus, the B+ decompositions corresponding to informationally-complete POVMs have a simple operational interpretation, and will be referred to as {\it canonical} B+ decompositions. Such canonical decompositions can be obtained from any complete set of system tomography observables, as shown in Appendix~\ref{condual}. An example based on a symmetric informationally-complete POVM (SIC-POVM) is given below, and generalised in Appendix~\ref{condual}. \blk

\begin{example}  \label{Ex2} (Canonical decomposition of qubit plus bath case). Consider the qubit SIC-POVM $\{P_\alpha=\frac{1}{4}({\bm 1}+  m^{(\alpha)} \cdot \sigma)\}$, defined via the unit Bloch vectors~\cite{sicpovm} 
\begin{eqnarray*}
m^{(0)} &=& \{0,0,1\}\\
m^{(1)} &=& \{\frac{2 \sqrt{2}}{3},0,-\frac{1}{3}\}\\
m^{(2)} &=& \{-\frac{\sqrt{2}}{3},\sqrt{\frac{2}{3}},-\frac{1}{3}\}\\
m^{(3)} &=& \{-\frac{\sqrt{2}}{3},-\sqrt{\frac{2}{3}},-\frac{1}{3}\}.
\end{eqnarray*}
These Bloch vectors form a regular tetrahedron, and the POVM elements satisfy $Tr[ P_{\alpha} P_{\alpha'}] = \frac{\delta_{\alpha,\alpha'}}{6} + \frac{1}{12}$. The dual frame is then given by $\{Q_\alpha=\frac{1}{2}( {\bm 1}+  \blk 3m^{(\alpha)} \cdot \sigma)\}$. \blk
\end{example}

\subsection{Connection with steering}

The bath states $\rho_\alpha$ appearing in a B+ decomposition, as per Eq.~(\ref{bplusgen}), are closely connected to the steering properties of the initial state $\rho_{SB}(0)$. In particular, if one measures some POVM $\{ E_{m}\}$ on the system, it follows from Eqs.~(\ref{maindeco}) and (\ref{rhoalpha}) that the bath is steered to the state $\rho'_m = \sum_\alpha w_\alpha (\tr{ E_m Q_\alpha}/p_m) \rho_\alpha$ for measurement outcome $m$, which occurs  with probability $p_m=\sum_\alpha w_\alpha {\rm Tr} [E_m Q_\alpha]$. Hence, the steered bath states are linear combinations of the $\rho_\alpha$, implying that the span of the set of steered states lies in the span of the $\rho_\alpha$ in the B+ decomposition.  

Note that if the Hilbert space of the system is $d$-dimensional, then choosing a canonical B+ decomposition as per Sec.~\ref{canonbplus} above yields at most $d^2$ different $\rho_\alpha$. Hence, the set of steered bath states must lie in a linear subspace of at most $d^2-1$-dimensions (applying the constraint that they must be normalised). If one considers the set of steered bath states for state $\rho_{SB}(t)$, i.e., as a function of time, then this linear subspace will in general also evolve over time.

\section{Dynamics of open quantum systems with arbitrary initial conditions}
\label{Dynamics}	

\subsection{A set of CPTP maps describes the reduced dynamics}
\label{aset}

Bath-positive decompositions have immediate consequences for representing the dynamics of open quantum systems, for the general scenario of initially correlated system-bath states. Given an arbitrary state at initial time, the reduced dynamics of the system at a time $t$ follows from Eqs.~(\ref{reducedintro}) and~(\ref{bplusgen}) as 
\beqy 
\label{gendynam}
\nonumber \rho_S(t) &=& \sum_\alpha w_\alpha {\rm Tr}_B[ U(t)(Q_\alpha\otimes \rho_\alpha)U(t)^\dagger]\\
& \label{reduced} =& \sum_\alpha w_\alpha \phi_t^{(\alpha)}(Q_\alpha),
\eneqy
where
\beq
\label{indCPTP}
\phi_t^{(\alpha)}(\cdot):= {\rm Tr}_B[ U(t)(\cdot \otimes \rho_\alpha)U(t)^\dagger]
\eneq
is a CPTP map acting on $\mathcal{H}_S$. Hence, in the general scenario, the system state evolution is described via a set of CPTP maps $\{\phi_t^{(\alpha)}\}$, each weighted by $w_\alpha$, and acting on the corresponding element $Q_\alpha$ of the basis. 

The  canonical construction of B+ decompositions in Sec.~\ref{canonbplus} implies that  $d^2$ CPTP maps are sufficient to describe the dynamics of a $d$-dimensional open quantum system.  {\blk We point out that while the maximum number of maps required is a consequence of the linearity of the Liouville equation, it is  by no means trivial to guarantee that the maps  are CPTP, as our decomposition does.} Moreover, \blk fewer maps can be sufficient in special cases. Indeed, in the factorisable case one has $\rho_{SB}(0) =\sum_\alpha w_\alpha Q_\alpha\otimes \rho_B$ from Eq.~(\ref{rhos}), yielding a fixed map $\phi_t^{(\alpha)}\equiv \phi_t$ for each $\alpha$, as expected from Eq.~(\ref{CPTPfact}). 

On the other hand, $d^2$ maps are in fact necessary for some initial system-bath states and interactions. For example, consider a pure entangled initial state $\rho_{SB}(0)=\ket{\Psi}\bra{\Psi}$ having Schmidt rank ${d}$, i.e., $\ket{\Psi} = \sum_{s=1}^{d}  a_s \ket{s} \ket{\chi_s}$ with $a_s>0$ and $\langle s|s'\rangle=\langle\chi_{s}|\chi_{s'}\rangle=\delta_{ss'}$. 
For any B+ decomposition as in Eq.~(\ref{bplusgen}), the basis elements $Q_\alpha$ can always be expanded relative to a basis $\{G_j\}$ of $d^2$ linearly independent operators satisfying $\tr{G_jG_k}=\delta_{jk}$, i.e., $Q_\alpha=\sum_j \tilde \T_{\alpha j}G_j$ for some real matrix $\tilde \T$ (see also Appendix~\ref{construct}). Hence,  evaluating $J_j:={\rm tr}_S[(G_j\otimes 1_B)\rho_{SB}(0)]$ via Eq.~(\ref{bplusgen}) yields
\[
J_j=\sum_{s,s'} a_s^* a_{s'} \bra{s} G_j \ket{s'}  \ket{\phi_{s'}}\bra{\phi_{s}}= \sum_\alpha w_\alpha \tilde\T_{\alpha j} \rho_\alpha .
\]
It follows from the first equality that $\sum c_j J_j=0$ if and only if $\sum_j c_j G_j=0$. Hence, since the $G_j$ are linearly independent, the $J_j$ also form a set of $d^2$ linearly independent operators. But from the second equality the $J_j$ are themselves linear combinations of the $\rho_\alpha$.  Thus, there must be no fewer than $d^2$ linearly independent $\rho_\alpha$ in the B+ decomposition (and hence, when \blk expanding in terms of a canonical B+ decomposition, there are exactly $d^2$ such $\rho_\alpha$). Correspondingly, it follows from Eq.~(\ref{indCPTP}) that there will typically be no fewer than $d^2$ linearly independent maps $\phi^{(\alpha)}_t$, providing the system-bath interaction is sufficiently nontrivial. Equivalently, in the absence of assumptions on the Hamiltonian ruling the evolution, $d^2$ linearly independent $\rho_\alpha$ will lead to to $d^2$ linearly independent maps $\phi^{(\alpha)}(\cdot)$. For example, if $U(t)$ corresponds to the swap operation at some time $t$, i.e., $U(t)(X\otimes Y)U(t)^\dagger=Y\otimes X$, for arbitrary $X$ and $Y$, then $\phi^{(\alpha)}_t(\cdot)=\tr{\cdot}\rho_\alpha$, and so the maps have precisely the same degree of linear independence as the $\rho_\alpha$ that generate them. 

Finally, we note there are  also  initial states  for which fewer than $d^2$ CPTP maps, \blk but more than one such \blk map,  are needed to describe the system dynamics.  As a first example, the argument of the preceding paragraph may be easily extended to show that a pure initial system-bath state with Schmidt rank $r$ requires no more than $r^2$ linearly independent CPTP maps.  A second example is provided by zero-discord initial states, for which~\cite{Discord}
\beq
\label{fac}
\rho_{SB}(0) = \sum_{\alpha=1}^d  w_\alpha \ket{\psi_\alpha}\bra{\psi_\alpha} \otimes \rho_\alpha,
\eneq
where the $\ket{\psi_\alpha}$ are mutually orthogonal. Noting that this already has the form of a B+ decomposition as per Eq.~(\ref{bplusgen}), it immediately follows  that at most $d$ independent dynamical maps are required to describe the evolution of zero-discord states. {\blk It is worth highlighting that for this very special case the dynamics for any choice of coefficients $w_\alpha$ summing to 1, i.e., for a $(d-1)$-dimensional space of states, is described by a single CPTP map, as shown in Ref.~\cite{CP1}. However, as will be shown in Sec.~\ref{implic}, the $d$ maps from the B+ decomposition describe the dynamics of a much larger space of initial states, having $2d(d-1)$ dimensions (obtained from $\rho_{SB}(0)$ via local operations on the system).  Finally, \blk we point out that a similar decomposition of zero discord states has been considered by Breuer, but with the roles of system and bath reversed~\cite{breuer07}. 

\subsection{Prediction and retrodiction in the presence of initial correlations}
\label{predret}

The most direct uses of the B+ decomposition come from its applicability to extend \blk techniques used to calculate the predicted dynamics of a quantum system under the factorisable initial state assumption, such as various master equation methods \blk \cite{ME1,ME2,ME3,ME4,ME5,ME6,ME7,ME8},  path integral methods~\cite{PI1,PI2,PI3,PI4}, and other techniques~ \cite{OpenQuantumSystemsReview, OpenQuantumSystemsBook2}. Recall that, in virtue of our decomposition, each of the maps $\phi_t^{(\alpha)}$ in Eq.~(\ref{indCPTP}) is  CPTP, since it originates from an initially uncorrelated operator $\rho^{(\alpha)}_{SB}:= Q_\alpha \otimes \rho_\alpha$. Typically, methods to compute the reduced dynamics {\it only} rely on the fact that ${\rm Tr}_S[\rho^{(\alpha)}_{SB}]$ is a valid density matrix but, crucially, make no stipulation about ${\rm Tr}_B [\rho^{(\alpha)}_{SB}]$.  In such cases, obtaining $\rho_S(t)$ is straightforward by applying such methods to each term in the decomposition and composing the outcomes as per Eq.~(\ref{reduced}). \blk The example of qubit dephasing is discussed in Sec.~\ref{dephasing} below. \blk Other methods, that require ${\rm Tr}_B[\rho^{(\alpha)}_{SB}]$ to satisfy particular properties, such as  purity, \blk can also be accommodated (see~\cite{LHWreview} for examples of ``Monte Carlo wave function simulations'' as we might call pure state techniques, both Markovian and non-Markovian). For example, the $Q_\alpha$ can always be expanded as a (not necessarily positive) linear combination of  projectors corresponding to pure states. It becomes then again a matter of solving each term in the expansion and combining the outcomes in the appropriate way. 

It should be highlighted that being able to predict the dynamics of a system using methods developed for CPTP maps is not the only interesting aspect. In fact, a similar argument can be used to {\it retrodict} the dynamics of the system, i.e., to estimate the density matrix of the system in the past, by computing $\phi_{-t}^{(\alpha)}$. Interesting questions, such as when were the system and bath in a factorisable state (if ever), can in principle be addressed. Obviously, retrodicting the dynamics of a state can also be done when the state is factorisable at time $t=0$, however doing so in that case is somewhat artificial. Being able to do so for an arbitrary state at $t=0$, as we can now, is certainly more natural. This is not merely an academic question requiring perfect knowledge of the system and bath. We will argue later in Section~\ref{QSense} that, under certain conditions, knowledge of $\phi^{(\alpha)}_{t>0}$ allows us to infer $\phi^{(\alpha)}_{t<0}$ and thus gives us the practical ability to retrodict the state of the system.  
 
\blk
\subsection{Example: Qubit dephasing for arbitrary initial correlations}
\label{dephasing}

As mentioned in Sec.~\ref{predret}, B+ decompositions allow one to immediately extend techniques used for solving the factorisable case to the general case. We demonstrate the power of this method here by showing how it may be used to fully solve an important model of quantum decoherence: qubit dephasing. In particular, we will show how  B+ decompositions allow the seamless extension of solutions devised for the factorisable case~\cite{massimo, reina, OpenQuantumSystemsBook2} to the general nonfactorisable case. 

The pure dephasing of a qubit coupled to a bosonic bath is described by the Hamiltonian~\cite{OpenQuantumSystemsBook2}
\beq
H = \half\epsilon\sigma_z + \sum_j \omega_jb^\dagger_jb_j + \sum_j g_j \sigma_z (b_j+b^\dagger_j) .
\eneq
Here $\epsilon$ is the qubit energy gap between
eigenstates of $\sigma_z$; $b_j$ and $b^\dagger_j$
are bath-mode annihilation and creation
operators, with corresponding frequencies $\omega_j$ and coupling strengths $g_j$ to the qubit; and units are such that $\hbar\equiv1$. This Hamiltonian provides a well-known energy-conserving model for qubit noise and decoherence, and the exact evolution of the qubit for the case of an initially uncorrelated  thermal bath is textbook material~\cite{OpenQuantumSystemsBook2}. More generally, the qubit evolution has been solved for all factorisable initial states~\cite{massimo} (and extended to multiple qubits~\cite{reina}). In contrast, the nonfactorisable case has been addressed explicitly for only a small set of initially correlated states~\cite{dajka}, and only implicitly for arbitrary initial states via the derivation of formal homogenous and inhomogenous master equations~\cite{ban}.

For a factorisable initial state $\rho_S\otimes \rho_B$, the diagonal elements of the qubit density operator with respect to the $\sigma_z$-basis are found to be constant in time~\cite{massimo, reina, OpenQuantumSystemsBook2}, i.e.,
\beq \label{dephasediag}
\langle0|\rho_S(t)|0\rangle = \langle0|\rho_S|0\rangle,~~~ \langle1|\rho_S(t)|1\rangle = \langle1|\rho_S|1\rangle,
\eneq
while the off-diagonal terms evolve in the interaction picture as per Eqs.~(4.9)--(4.13) of~\cite{massimo}, with
\beq \label{dephasefact}
\langle0|\rho_S(t)|1\rangle = \langle0|\rho_S(t)|1\rangle ~{\rm Tr}_B[\rho_B D(\bm\xi_t)] .
\eneq
Here $\bm \xi_t=(\xi_1(t),\xi_2(t),\dots)$, with 
\beq
\xi_j(t) := 2g_j \frac{1-e^{i\omega_jt}}{\omega_j} ,
\eneq
and $D(\bm\xi)=\exp(\sum_j \xi_j b^\dagger_j-\xi^*_jb_j)$ denotes the multimode Glauber displacement operator. Note that the scaling factor ${\rm Tr}_B[\rho_B D(\bm\xi_t)]$ in Eq.~(\ref{dephasefact}) is the characteristic function corresponding to the Wigner function of the initial bath state $\rho_B$~\cite{wigner,gauss}, and hence this equation may be rewritten as
\beq \label{dephasechi}
\langle0|\rho_S(t)|1\rangle = \langle0|\rho_S(t)|1\rangle ~\chi_{\rho_B}(\bm\xi_t).
\eneq
The important case of Gaussian bath states, including coherent, squeezed and thermal states, is characterised by $\chi_{\rho_B}(\bm\xi)$ being a Gaussian function of $\bm\xi$~\cite{gauss}, and thus qubit dephasing is particularly simple for such bath states~\cite{OpenQuantumSystemsBook2, massimo,reina}.

The qubit evolution for the general case of a correlated initial state $\rho_{SB}$ is now easily determined by the method of B+ decompositions.  First, choose any convenient B+ decomposition 
\beq
\rho_{SB}=\sum_\alpha w_\alpha Q_\alpha\otimes \rho_\alpha
\eneq 
as per Eq.~(\ref{bplusgen}), e.g., as per either of Examples~\ref{Ex1} and~\ref{Ex2} in Sec.~\ref{theory}. Second, let $\chi_{\rho_\alpha}(\bm \xi)$ denote the Wigner characteristic function of bath state $\rho_\alpha$ in this decomposition.  Defining $Q_\alpha(t):= \phi_t^{(\alpha)}(Q_\alpha)$, it  follows immediately from Eqs.~(\ref{indCPTP}, (\ref{dephasediag}) and~(\ref{dephasechi}) that
\beq
\langle0|Q_\alpha(t)|0\rangle = \langle0|Q_\alpha|0\rangle,~~~
 \langle1|Q_\alpha(t)|1\rangle = \langle1|Q_\alpha|1\rangle ,
\eneq
\beq 
\langle0|Q_\alpha(t)|1\rangle = \langle0|Q_\alpha|1\rangle ~\chi_{\rho_\alpha}(\bm\xi_t).
\eneq
Finally, substituting these results into $\rho_S(t)=\sum_\alpha w_\alpha Q_\alpha(t)$ as per Eq.~(\ref{gendynam}), the general qubit evolution is given by
\beq \label{dephasediag2}
\langle0|\rho_S(t)|0\rangle = \langle0|\rho_S|0\rangle,~~~ \langle1|\rho_S(t)|1\rangle = \langle1|\rho_S|1\rangle,
\eneq
and 
\beq \label{dephasechi2}
\langle0|\rho_S(t)|1\rangle = \sum_\alpha w_\alpha \langle0|Q_\alpha|1\rangle ~\chi_{\rho_\alpha}(\bm\xi_t),
\eneq
generalising Eqs.~(\ref{dephasediag}) and~(\ref{dephasechi}) for the uncorrelated case. Note that the diagonal elements are constant for the general case, similarly to the uncorrelated case, while the off-diagonal elements are a weighted sum of the Wigner characteristic functions of the bath states $\rho_\alpha$. 

The above example shows that using B+ decompositions provides a straightforward mechanism for turning the problem of correlated initial states into the problem of factorisable initial states, and hence allowing the general case to be solved via known methods for the factorisable case. Applications of B+ decompositions going beyond known methods for factorisable states will be discussed in Secs.~\ref{retrosec} and~\ref{LAT}. \blk

\blk Lastly, it is worth remarking that, in analogy to the factorisable case, Eq.~(\ref{dephasechi2}) is particularly simple for {\it Gaussian} B+ decompositions, in which the bath states $\rho_\alpha$  are Gaussian states. Such Gaussian B+ decompositions are expected to provide useful approximations for the qubit evolution in the case that $\rho_{SB}$ is close to the product of the initial system state with a Gaussian bath state, such as a thermal bath state. {\blk In this case, measurement of a POVM $\{P_\alpha\}$ as per the Example~\ref{Ex2} in Sec.~\ref{canonbplus} will typically extract little information about the bath state, so that the corresponding bath state $\rho_\alpha$ in Eq.~(\ref{rhoalpha}) will remain close to a Gaussian state, and hence can be well approximated by a Gaussian state having the same mean and covariance properties as $\rho_\alpha$. This is a worthwhile topic for future investigation.

\blk
\subsection{Extended applicability \blk under local system operations}
\label{implic}

As remarked in the Introduction, the use of a single CPTP map or quantum channel, to describe open system evolution, is restricted to a very small class of initial system-bath states~\cite{CP1,CP2,CP3,CP4,CP5,CP6,CP7,CP8,CP9,CP10,CP11,CP12}.	For example, for a $d$-dimensional system, the dynamical map $\phi_t$ in Eq.~(\ref{CPTPfact})  for an uncorrelated \blk bath state $\rho_B$ only applies to the evolution of a $(d^2-1)$-dimensional space of factorisable initial states, of the form $\rho_S\otimes \rho_B$.  

Here we demonstrate that, in contrast, the set of dynamical maps $\{\phi^{(\alpha)}_t\}$ in Eq.~(\ref{indCPTP}) \blk may be applied \blk to the evolution of a much larger space of initial states, having up to $d^4-1$ dimensions. This space corresponds to \blk precisely those states that can be prepared from $\rho_{SB}(0)$ via  local operations (including measurements) \blk on the system. Thus, while up to $d^2$ CPTP maps may be needed for describing dynamics of an initially correlated state, they have the corresponding predictive advantage of describing the evolution of up to a $(d^4-1)$-dimensional space of operationally-related initial states. \blk For the qubit dephasing example in Sec.~\ref{dephasing}, this corresponds to being able to use the same 4 characteristic functions $\chi_{\rho_\alpha}(\bm \xi)$ to solve for the evolution of a 15-dimensional space of initially-correlated states. In contrast, in the uncorrelated case the single CPTP map $\phi_t$ only provides the solution for a 3-dimensional space of initial states.  

In particular, let $\R$ be a CP linear \blk map acting on the system, corresponding to some operation. It may be trace-decreasing, in which case it corresponds to a measurement, 
also called a selective or filtering \blk operation  For example, $\R_\psi(X):=|\psi\rangle\langle\psi|X|\psi\rangle\langle\psi|$ describes  an ideal projective measurement that projects $\rho_S$ onto ket $|\psi\rangle$ with probability $p_\psi= \langle\psi|\rho_s|\psi\rangle$.  More generally, a system map $\R$ acting on the system-bath state $\rho_{SB}$ will prepare it in a state $\rho^\R_{SB}$ with probability $p_\R$, defined implicitly via
\beq \label{rhor}
p_\R\, \rho^\R_{SB}:=(\R\otimes I_B)(\rho_{SB}) . 
\eneq
Here $I_B$ denotes the identity map on the bath. The class of system-bath states that can be prepared in this way, with $p_\R>0$, will be denoted by  ${\cal S}_{\rho_{SB}}$. Note that this class includes all factorisable states $\tau\otimes\rho_B$ as a special subclass, where $\rho_B={\rm Tr}[\rho_{SB}]$, since such states are generated by the corresponding local replacement operations $\R_{\tau}(X):=\tau\,{\rm Tr}_S[X]$. Factorisable  states  are in fact the only states that can be prepared from a factorisable initial state $\rho_{SB}=\rho_S\otimes\rho_B$,  via local system operations.  More typically, however, ${\cal S}_{\rho_{SB}}$ contains many further states, as shown in more detail below.

We now \blk show that the set of dynamical maps $\{\phi_t^{(\alpha)}\}$ in Eq.~(\ref{indCPTP}) not only determines the system evolution for the initial state $\rho_{SB}(0)$, but for the entire class of states in ${\cal S}_{\rho_{SB}(0)}$: 

\begin{theorem} 
\label{Th1}
	The system dynamics for each member of the class of initial states ${\cal S}_{\rho_{SB}(0)}$, obtained from $\rho_{SB}(0)$ by performing local operations on the system, is determined by the single set of dynamical maps $\{\phi_t^{(\alpha)}\}$ defined in Eq.~(\ref{indCPTP}).  
\end{theorem}
\begin{proof}
Note first from Eq.~(\ref{bplusgen}) that
\begin{align} \label{rhorsb}
p_{\R}\,\rho^\R_{SB}(0) & =\sum_\alpha w_\alpha \mathcal{R}(Q_\alpha) \otimes \rho_\alpha=\sum_{\alpha,{\alpha'}} w_\alpha R_{\alpha{\alpha'}} \,Q_{\alpha'}\otimes \rho_\alpha ,
\end{align}
where $\sum_{\alpha'}R_{\alpha{\alpha'}}Q_{\alpha'}$ is any expansion of the system operator $\R(Q_\alpha)$ with respect to basis $\{Q_\alpha\}$ (from Eq.~(\ref{decom}) we may take $R_{\alpha\alpha'}={\rm Tr}[\R(Q_\alpha)P_{\alpha'}]$). Taking the trace over the bath yields $p_\R=\sum_\alpha w_\alpha {\rm Tr}[\R (Q_\alpha)]$, and hence the subsequent system evolution is determined by the maps $\{\phi_t^{(\alpha)}\}$ via
\begin{align} \label{proof1}
\rho^\R_S(t)	& := {\rm Tr}_B[\rho^\R_{SB}(0)]= \frac{\sum_{\alpha,{\alpha'}} w_\alpha R_{\alpha{\alpha'}}\, \phi^{(\alpha)}_t(Q_{\alpha'})}{\sum_\alpha w_\alpha {\rm Tr}[\R (Q_\alpha)]} ,
\end{align}
as required. Note that in the special case where no operation is performed, i.e., $\R=I_S$, Eq.~(\ref{proof1}) reduces to Eq.~(\ref{reduced}) for $\rho_S(t)$. 
\end{proof}

To determine the size of ${\cal S}_{\rho_{SB}(0)}$, note that the set of local system operations, $\{\R\}$, has $d^4$ real degrees of freedom for the case of a $d$-dimensional system. Hence, noting the constraint $\tr{\rho^\R_{SB}}=1$, the set of initial states ${\cal S}_{\rho_{SB}}$ prepared from a given state $\rho_{SB}$ via such local operations spans at most $(d^4-1)$ dimensions. The maximum is reached, for example, for any pure initial state $\rho_{SB}(0)=\ket{\Psi}\bra{\Psi}$ having Schmidt rank $d$ (see Sec.~\ref{aset}), noting that the mapping $\R\rightarrow (\R\otimes I_B)(\ket{\Psi}\bra{\Psi})$ is linear and 1:1 for such states. 

\blk As a further example, note for the zero-discord initial state in Eq.~(\ref{fac}) that the local operation $\R(X)=\sum_\alpha |\phi_\alpha\rangle\langle\psi_\alpha| X|\psi_\alpha\rangle\langle\phi_\alpha|$ maps  the $d$ orthogonal system states $\psi_\alpha$ to $d$ arbitrary system states $|\phi_\alpha\rangle$. Hence, since a pure state is described up to normalisation and a global phase by $2d-2$ parameters, it follows that the set of initial states ${\cal S}_{\rho_{SB}} =\{ \sum_{\alpha=1}^d  w_\alpha \ket{\phi_\alpha}\bra{\phi_\alpha} \otimes \rho_\alpha\}$, has $2d(d-1)$ dimensions, with its evolution described by $d$ CPTP maps. \blk

\subsection{Relation to superchannels, process tensors, and process tomography} 
\label{sch}

In Ref.~\cite{modi}, Modi has recently suggested representing the evolution of an open quantum system by a superchannel, ${\cal C}_t$, which maps the set of trace-preserving local system operations to the state of the system at time $t$ (see also~\cite{modi1,modi2} for very interesting generalisations of this approach). Thus, in our notation, $\rho^\R_S(t)={\cal C}_t(\R)$, where $\R$ is now restricted to be a CPTP map, i.e., \blk with $p_\R=1$. It immediately follows from Eq.~(\ref{proof1}) that this superchannel can be explicitly represented in terms of the CPTP maps $\{\phi_t^{(\alpha)}\}$, via
\beq
\label{superch}
 \rho^\R_S(t)={\cal C}_t(\R) = \sum_{\alpha,{\alpha'}} w_\alpha R_{\alpha{\alpha'}}\, \phi^{(\alpha)}_t(Q_{\alpha'}) .
\eneq 
Thus, B+ decompositions lead to corresponding decompositions of the superchannel (with simple operational interpretations in the case of canonical B+ decompositions). {The superchannel description of open quantum system dynamics and its generalisations are a powerful tool. However, an advantage of using a representation of system dynamics based on the $\{\phi_t^{(\alpha)}\}$, rather than dealing directly with the superchannel, is that the problem of determining the system evolution is reduced to the consideration of CPTP maps acting directly on system operators, for which many tools exist~\cite{OpenQuantumSystemsReview,OpenQuantumSystemsBook1,OpenQuantumSystemsBook2} (see also Secs.~\ref{aset} and~\ref{predret}). \blk Further, the equivalence in Eq.~(\ref{superch}) implies  that methods developed in the context of superchannels can be directly translated into the familiar language of CPTP  system \blk maps via the judicious application of our results.  }

For example, \blk while it is possible to experimentally determine the superchannel ${\cal C}_t$ at a given time $t$, in terms of a basis set of CPTP maps $\{\R^{(j)}\}$~\cite{modiexp}, this is equivalent to experimentally determining the maps $\phi^{(\alpha)}_t(\cdot)$ at time $t$. In particular, applying the basis map $\R^{(j)}$ on the system at time $0$ and evolving to time $t$ allows the $\rho_S^{\R^{(j)}}(t)$ to be tomographically reconstructed~\cite{PT1, PT2}. Repeating this procedure for each basis map, Eq.~(\ref{superch}) then yields the set of linear equations
\beq 
\rho^{\R^{(j)}}_S(t) = \sum_{\alpha, \alpha',\beta} R^{(j)}_{\alpha\alpha'} F^{(\alpha)}_{\alpha'\beta}(t) Q_\beta ,
\eneq
where $F^{(\alpha)}_{\alpha'\beta}(t):=w_\alpha\tr{ \phi^{(\alpha)}_t(Q_{\alpha'})P_\beta}$, and  Eq.~(\ref{decom}) has been used with $A=\phi^{(\alpha)}_t(Q_{\alpha'})$. Choosing a canonical B+ decomposition for convenience, so that the $Q_\beta$ are linearly independent, as are the representations $R^{(j)}_{\alpha\alpha'}$ of the basis maps, this determines $F^{(\alpha)}_{\alpha'\beta}(t)$ uniquely. The dynamical map $\phi^{(\alpha)}_t$ can be determined via its action on the basis operators $\{Q_\alpha\}$:
\beq
\phi^{(\alpha)}_t(Q_{\alpha'}) = (w_\alpha)^{-1}\sum_\beta F^{(\alpha)}_{\alpha'\beta}(t) Q_\beta
\eneq
(recall that $w_\alpha$ is determined by the initial state $\rho_S(0)$). A different approach is taken in Sec.~\ref{QSense}, where a generalised quantum sensing protocol is introduced to effectively estimate the $\phi^{(\alpha)}_t$ for a continuous range of times $t$, rather than at a single time.
\blk

Superchannels can be generalised to the scenario where multiple local operations are applied at times $t_i$ of the evolution, $\mathcal{R}^{(i)}_{t_i}$. This is achieved using the so-called process tensor, introduced in Refs.~\cite{modi, modi1,modi2,costa}, that is a completely positive supermap taking the local operations $\mathcal{R}^{(i)}_{t_i}$ to a state $\rho_S(T)$. It is also possible to generalise our formalism in this direction. Applying $N+1$ local operations each followed by a system bath unitary $U^{(i)}$, one has 
\begin{align*}
\rho_S(T) &=\textrm{Tr}_B [ U^{(N)}\cdot ( \mathcal{R}^{(N)}_{t_{N}}(\cdots (U^{(0)} \cdot (\mathcal{R}^{(0)}_{t_{0}}\cdot (\rho_{SB}(0)) \cdots )]\\
&= \sum_{\vec{\beta},\vec{\beta'}} w_{\beta_0} R^{(0)}_{\beta_0,\beta'_0} R^{(1)}_{\beta_1,\beta'_1} \cdots R^{(N)}_{\beta_N,\beta'_N} \times  \textrm{Tr}[ P_{\beta_1} \phi^{(\beta_0)}_0 (Q_{\beta'_0})] \\
&\,\,\,\,\,\,\, \textrm{Tr}[ P_{\beta_2} \phi^{(\beta_1)}_1 (Q_{\beta'_1})] \cdots \textrm{Tr}[ P_{\beta_{N+1}} \phi^{(\beta_{N})}_N (Q_{\beta'_{N}})]  Q_{\beta_{N+1}}
\end{align*}
where $U\cdot X = U X U^\dagger$, $\textrm{Tr}_B[ U^{(i)}  Q_\alpha \otimes \rho_\alpha {U^{(i)}}^\dagger] =  \phi^{(\alpha)}_i ( Q_\alpha)$ and we have used the observation that 
$$
U^{(i)} (Q_\alpha \otimes \rho_{\alpha'}){U^{(i)}}^\dagger = \sum_{\gamma} \tr{P_\gamma \phi^{(\alpha')}_i (Q_\alpha)}\,\, Q_\gamma \otimes \rho'_\gamma,
$$  where $\rho'_{\gamma} = \textrm{Tr}_S [(P_\gamma\otimes I_B) (U^{(i)} (Q_\alpha \otimes \rho_{\alpha'}){U^{(i)}}^\dagger )]$. 
The above expression shows that, in direct analogy to the process tensor, one can write the final state $\rho_S(T)$ as a function of the components $R^{(i)}_{\beta,\beta'}$ which fully determine $\mathcal{R}^{(i)}$, while all the information about the dynamics an initial state is stored in the $p_{\beta_0}$ and the components $\{f^{i,\beta_i}_{\beta_{i+1},\beta'_i} \equiv \textrm{Tr}[ P_{\beta_{i+1}} \phi^{(\beta_i)}_i (Q_{\beta'_i})]\}$ uniquely determining the CPTP maps $\phi_i^{\beta_i}$. As in the scenario of the single operation, and much in the same way that is done for process tensors, by cycling over an approrpiate set of operations $\mathcal{R}^{(i)}$ it is possible to recover the strings $\{f^{N,\beta_N}_{\beta_{N+1},\beta'_N} \cdots f^{0,\beta_0}_{\beta_{1},\beta'_0}\}$ and thus reconstruct the full expression for $\rho_S(T)$.

\subsection{Computational Markovianity} 
\label{sec:dm}

The notion of Markovian (i.e., memoryless) dynamics for an open quantum system has been the object of intense study, and is at the heart of many common methods to describe open quantum systems~\cite{OpenQuantumSystemsBook1,OpenQuantumSystemsBook2}. Many popular definitions of quantum Markovianity, such as decreasing state distinguishability~\cite{BLP2009} and divisibility~\cite{RHP2010}, are written in terms of a dynamical map describing the evolution of the system. Thus, these definitions implicitly assume the factorisable initial state condition. Indeed, it might be argued that the evolution of a system initially correlated with its environment is necessarily non-Markovian, as a `memory' of the initial system state could be propagated via these correlations.

However, our observation is that even though the evolution of the system state $\rho_S(t)$ will in general depend  on the initial joint state $\rho_{SB}(0)$, the dynamics itself can be memoryless, in the sense that in Eq.~\eqref{reduced} the evolution induced by the system bath unitary dynamics and each of the $\rho_\alpha$ can be memoryless (or Markovian) according to one or more of the definitions used for initially factorisable states~\cite{RHPreview, Breuerreview,LHWreview}. In such  a case, any memory of the initial state is \blk only encoded in the initial weights $\{w_\alpha\}$ in Eq.~(\ref{rhos}).  In particular, none of the maps $\phi^{(\alpha)}_t$ describing the dynamics has any such memory. Thus, Markovian methods can be used to calculate $\rho_S(t)$ in this  case, by \blk solving  up to \blk $d^2$ parallel Markovian dynamical equations.  This motivates the following: 

\begin{definition} \label{defdm}{\rm (computational Markovianity)} {\it Let ${\cal M}[H(t),\rho_B]$ be some definition or criterion of Markovianity, for the evolution of a system induced by system-bath Hamiltonian $H(t)$ \blk and an initially uncorrelated bath state $\rho_B$. We then say that the corresponding evolution for an initially correlated system-bath state $\rho_{SB}$ is computationally Markovian, relative to this definition or criterion, if and only if there exists a B+ decomposition of $\rho_{SB}$ such that ${\cal M}[H(t),\rho_\alpha]$ holds for all $\alpha$ with $w_\alpha\neq 0$.}
\end{definition}

{\blk We stress that, as well as allowing a generalisation of the many different notions of Markovianity in the literature~\cite{RHPreview, Breuerreview,LHWreview} to the non-factorisable scenario, our definition of computational Markovianity has mainly a practical motivation: the dynamics of an initially correlated system is computationally Markovian if it can be faithfully calculated by using Markovian methods.}

Note that the definition of computational Markovianity applies not only to cases where ${\cal M}(H(t),\rho_B)$ can be formulated solely in terms of the dynamical map~\cite{RHPreview, Breuerreview}, but also to cases where ${\cal M}(H(t),\rho_B)$ depends explicitly on properties of $H(t)$ and/or $\rho_B$ (e.g., for definitions corresponding to quantum white noise or the quantum regression formula)~\cite{LHWreview}.  The above  definition  implies, as should be expected, that if computational Markovianity holds for some  pair $(H(t), \rho_{SB})$, then it also holds for any pair  $(H(t), \rho^{\cal R}_{SB})$,  where $\rho^{\cal R}_{SB}$ obtained from $\rho_{SB}$ via a local system operation $\cal R$ as per Eq.~(\ref{rhor}).  In particular,  replacing $Q_\alpha$ by $\R(Q_\alpha)$ yields a B+ decomposition of $\rho^{\cal R}_{SB}$ having the same $\{w_\alpha\}$ and $\{\rho_\alpha\}$.  \blk It would be of interest to explore computational Markovianity for the Gaussian $B+$ decompositions discussed in Sec.~\ref{dephasing} for qubit dephasing, and compare with corresponding results for the factorisable case~\cite{haikka,giacomo}.

While it is relatively easy to check whether a given B+ decomposition of the initial state is Markovian relative to \blk a given definition ${\cal M}(H(t),\rho_B)$, \blk it is harder to determine whether the evolution is computationally Markovian for a given $\rho_{SB}(0)$. \blk That is because determining this \blk potentially involves searching the space of B+ decompositions. In Appendix~\ref{DMark} we provide a protocol that requires a fixed number of measurements (necessary to characterise the evolution relative to one B+ decomposition), and does the search via classical processing. 
	
{In some cases, the question of computational Markovianity for the evolution is no more difficult than the question for a fixed B+ decomposition. A} prime example of this situation is verifying computational Markovianity for the  criterion which associates the complete failure of dynamical decoupling (DD) with Markovian dynamics~\cite{LHWreview}. In this scenario, if there exists a frame such that the evolution induced by $H(t), \rho_\alpha$ is invariant under the action of DD pulses, then in this frame {\it all} the $\phi_t^{(\alpha)}(\cdot)$ resulting from any such evolutions are invariant under DD. Now, \blk in a different frame, the associated maps ${\phi'}_t^{ (\alpha)}(\cdot)$ are necessarily linear combinations of the original frame, i.e., ${\phi'}_t^{\alpha}(\cdot) = \sum_{\beta} C_{\beta,\alpha} \phi_t^{\alpha}(\cdot)$. Hence, \blk invariance of the 
$\phi_t^{(\alpha)}(\cdot)$ necessarily implies invariance of the  ${\phi'}_t^{(\alpha)}(\cdot)$, i.e., the evolution induced by the maps associated with {\it any} B+ decomposition is automatically also invariant. Thus computational Markovianity, in the sense of the DD criterion, is satisfied in one frame if and only if it is satisfied in every frame.

\section{Completing the picture:\\  Extended Quantum Sensing Protocols for arbitrary initial conditions.}
\label{QSense}

In recent years~\cite{Spectro1,Spectro2,Spectro3,Spectro4,Spectro5,Spectro6,Spectro7,Spectro8,Spectro9}, bath correlations have been recognised as the key to satisfying the {\it sufficient dynamical information} condition in Sec.~\ref{introduction} \blk and to achieve high fidelity operations, when the bath cannot be directly accessed. Thus, noise spectroscopy protocols of varying generality have been developed. Their limitations come either in terms of restrictions to mathematically amenable noise models or in terms of the the detail and quality of information that can be obtained about the bath correlations~\cite{Spectro1,Spectro2,Spectro3,Spectro4,Spectro5,Spectro6,Spectro7,Spectro8,Spectro9}. While a complete solution, i.e., capable of reconstructing the bath correlations to arbitrary order with minimal assumptions on the noise model, is the objective of current efforts, this principle applies more generally: a quantum system can be used to extract information about the bath it interacts with. Indeed, the use of quantum probes to extract information in this way is the purview of general {\it quantum sensing protocols} (QSPs), recently reviewed  in~\cite{QSensing}, \blk of which noise spectroscopy is perhaps the most ambitious example. Existing sensing protocols have so far been studied in the factorisable regime, and can be formally defined as follows. 

\begin{definition} \label{defsp}  {\rm (Quantum sensing protocols)} Let the dynamics of a {quantum probe {$S$} and a probed system} $B$ \blk be ruled by the Hamiltonian $H(t) = \sum_{b} V_b \otimes B_b(t)$, and let the initial state be of the form $\rho_{SB}(t=0) = \rho_S \otimes \rho_B$. A quantum sensing protocol consists of a set of (possibly adaptively chosen)
\begin{itemize}
\item[(i)] initial system states $\{\eta_\mu\}$,
\item[(ii)] system control Hamiltonians $H_{\beta}(t)$ ,
\item[(iii)] system observables $\{O_\gamma\}$, and
\item[(iv)] a classical processing routine \texttt{Sense}. 
\end{itemize}
The routine \texttt{Sense} takes as inputs the expectation values $E_{\eta_\mu,H_\beta,O_\gamma,\rho_B} = {\rm Tr} [  U^{(\beta)}(\eta_\mu \otimes \rho_B) {U^{(\beta)}}^\dagger O_\gamma]$, with $U^{(\beta)}$ the evolution generated by $H(t)+ H_\beta(t)$ from $t=0$ to a time $t=T$ that may depend on the choice of $H_\beta(t)$. Its output {\rm OUT} is some desired information about parameter(s) of the bath or probed system. \blk We will say that the protocol is {\em restricted} when, given $H(t)$, assumptions on $\rho_B$ are needed for \texttt{Sense} to be well-defined, i.e., to give accurate reconstructions in principle, and \emph{generic} when $\rho_B$ can be arbitrary.
\end{definition}

The above definition of QSPs \blk should be understood to include \blk the scenario when the quantity to be sensed is classical in nature, i.e., the probe couples to an stochastic parameter one wants to characterise. In this case, the `quantum' average with respect to the state of the bath, $\langle \cdot \rangle = \tr{\cdot \rho_B}$, is replaced by the classical mean, e.g., averaging over realisations of the stochastic process. This broad definition encompasses various well-known applications~\cite{QSensing}. On one side of the spectrum, traditional phase estimation protocols~\cite{PhEst1} can be seen as sensing routines for estimating a constant process $B(t)=B_0$. At \blk the other end, general noise spectroscopy protocols represent the most ambitious version of sensing as it seeks to characterise the dominant correlation functions $\langle B_{b_1}(\omega_1) \cdots B_{b_k}(\omega_k)\rangle_{k \leq K}$, for some finite $K$, of an arbitrary noise process $B_{b_1}(\omega_1)$, where $B_j(\omega)$ denotes the Fourier transform of $B_j(t)$. 

{\blk While extremely useful, and indeed one of the cornerstones of quantum technologies, a limitation of current quantum sensing protocols is their reliance on the initially factorisable condition. In many scenarios this can be a reasonable assumption, but it is by no means guaranteed that a probe  and the  probed system are initially uncorrelated. This is particularly true in the context of noise spectroscopy protocols and their application to high fidelity control. Moreover, it may be, for example, that the information of interest is encoded in the initial correlations.  In Section~\ref{LAT}, we provide an example where we showcase the importance of extending QNS (particularly quantum noise spectroscopy) protocols to arbitrary initial conditions. Thus, it is of general interest to extend QSP's to arbitrary initial conditions. Let us see how it can be done.}

\begin{theorem} \label{thm2}
	Any QSP defined in the $\rho_{SB} = \rho_S \otimes \rho_B$ scenario can be extended to the correlated \blk initial condition scenario, i.e., $\rho_{SB} \neq \rho_S \otimes \rho_B$, by adding the ability to perform system-only CP operations at $t=0$.

We will call such a protocol an {\it extended} QSP (eQSP). \blk
 \end{theorem}

 \begin{proof}
 	We need to show that, for any given arbitrary initial state  $\rho_{SB}(0)$, \blk observable $O_\gamma$, and control Hamiltonian $H_\beta(t)$, it is possible to obtain the expectation value $E_{\eta_\mu, H_\beta,O_\gamma,\rho_\alpha}$. for any desired $\eta_\mu$ and all $\rho_\alpha$  in a B+ decomposition  $\rho_{SB}(0) = \sum_{\alpha} w_\alpha Q_\alpha \otimes \rho_\alpha$ of the initial state\blk---despite not having access to an initial state of the form $\eta_\mu \otimes \rho_\alpha$. If this can be done, \blk then the  QSP \blk is directly applicable to each of the $\rho_\alpha$ associated with the B+ decomposition for the basis \blk $\{ P_\alpha\}$, and  also (via linearity) \blk to any $\rho'_\alpha$ resulting from a  different decomposition associated to a different set $\{P'_\alpha\}$. The output of  an  extended generic QSP \blk (e.g.~QSP)  is then \blk the information OUT  for every $\alpha$,  e.g., \blk a functional of $\langle B_{b_1}(\omega_1) \cdots B_{b_k}(\omega_k)\rangle_\alpha = {\rm Tr} [ B_{b_1}(\omega_1) \cdots B_{b_k}(\omega_k) \rho_\alpha]$. For a restricted quantum sensing protocol (rQSP), the corresponding  extended protocol (erQSP) \blk requires that each of the $\rho_\alpha$ satisfies the conditions of the  rQSP. \blk

Concretely, suppose \blk that at $t=0$  the preparation stage outputs a state $\rho_{SB}(0)$. Since the basis $Q_\alpha$ is known, doing tomography on the system state provides us with the $\{ w_\alpha\}$ via Eq.~(\ref{weights}). Then, one can apply a local, i.e., system-only, operation $\mathcal{R}^{(j)}: Q_\alpha \rightarrow  \sum_{\alpha'} R^{(j)}_{\alpha,\alpha'} Q_{\alpha'}$ before the evolution takes place, as in Sec.~\ref{implic}. In this way, for a choice of control Hamiltonian $H_\beta(t)$ and of observable $O_\gamma$, at the end of the experiment the expectation value
\begin{align}
\nonumber {\rm Tr}[ O_\gamma \rho^{(j)}_{S} (T) \blk] &= \sum_\alpha w_\alpha {\rm Tr}[ O_\gamma U^{(\beta)}(\mathcal{R}^{(j)}(Q_\alpha) \otimes \rho_\alpha) {U^{(\beta)}}^\dagger]\\
\label{expects} &=  \sum_{\alpha,\alpha'} w_\alpha R^{(j)}_{\alpha, \alpha'} E_{Q_{\alpha'}, H_\beta,O_\gamma, \rho_\alpha},
\end{align}
with $\rho^{(j)}_{S} = {\rm tr} [(\R^{(j)} \otimes I_B)(\rho_{SB})] $, can be calculated from measurable quantities.  Notice that for a given $\mathcal{R}^{(j)}$, this is a linear function of the variables $E_{Q_{\alpha'}, H_\beta,O_\gamma, \rho_\alpha}$, with known coefficients $w_\alpha R^{(j)}_{\alpha, \alpha'}$. For fixed $O_\gamma$ and $H_\beta(t)$, it is then possible to construct, using a similar argument to the one in ~\ref{sch}, a suitable set of $\{\mathcal{R}^{(j)}\}$ such that the $\{{\rm Tr}[ O_\gamma \rho^{(j)}_{S}]\}$ form an invertible linear set of equations from which all the $E_{Q_{\alpha'}, H_\beta,O_\gamma, \rho_\alpha}$ can be obtained. Notice that access to arbitrary $\R^{(j)}$ can be achieved via the use of an ancillary system and joint evolution/measurements (see for example Ref.~\cite{CPTPGen}). 

It follows, since  $\{Q_\alpha\}$ is a basis set, that one can calculate $E_{\eta_\mu, H_\beta,O_\gamma, \rho_\alpha}$ for an arbitrary $\eta_\mu$. Therefore, by repeating the above process for the appropriate set of observables and control Hamiltonians, one can effectively \blk apply the QSP \blk to each $\rho_\alpha$ independently, despite never preparing an initial state of the form $\rho_S \otimes \rho_\alpha$ as per the theorem. 
\end{proof}

\blk In terms of noise spectroscopy and the dynamics of open quantum systems in the presence of initial correlations, this implies that it is in principle possible to reconstruct $\{\langle B_{b_1}(\omega_1) \cdots  B_{b_k}(\omega_k)\rangle_\alpha \}$ for every $\rho_\alpha$.  Thus it is possible,  given an arbitrary initial state, to calculate each of the $\phi^{(\alpha)}_t(\cdot)$ ruling the dynamics, from measurable quantities. \blk Moreover, we highlight that, in virtue of Theorem~\ref{Th1},  characterizing the correlators ruling the dynamics of a state $\rho_{SB}$, and thus the corresponding set of CPTP maps, implies that one has direct access to the correlators describing the dynamics of any state related  $\rho_{SB}$  by a system-only operation. \blk

\section{Application: practical retrodiction of system dynamics}
\blk 
\label{retrosec}

Interestingly, being able to solve the dynamics of an arbitrary initial state makes meaningful the question of {\it retrodicting} the dynamics of a system (see also Sec.~\ref{predret}). \blk That is, imagine that at time $T_0\leq 0$ a preparation procedure \blk is executed on the system and bath, \blk resulting in a state $\rho_{SB}(T_0)$.  Then the system and bath evolve under a Hamiltonian $H(t)= H_S + H_{SB} + H_B$ from $t=T_0$ to $t=0$, at which point we are given access and are allowed to apply controls and to measure the system. We assume at this point that $H_S$ is known. As discussed above, we can use tomography and noise spectroscopy tools to learn about {the leading bath correlators within a time interval $t \in [0,T]$}, which will eventually lead to the capacity to predict and control its expected dynamics within the interval with high accuracy.  However, can we specify the state of a system at times $t=T_{-} \leq 0$  and,  in particular, the state resulting from the measurement or preparation procedure at time $T_0$? That is what we mean by retrodiction. 
	
{ The possiblity of such retrodiction has been previously considered using a master equation approach, for factorisable states at $t=0$, under a strong Markovian assumption~\cite{Retro1}.  Here we show how retrodiction may be achieved under much weaker conditions, and in particular without assumptions of the system and bath state being factorisable at any time or of Markovianity. }

To proceed, note first that the \blk retrodicted state $\rho_S(T_-)$ that we seek,  i.e., the system density operator we would have measured if we had done tomography at time $t=T_-\leq 0$, is given by
\beq
\label{gendynamretro}
\rho_S(T_-) = {\rm Tr}_B [U(T_{-}) \rho_{SB}(0) U^\dagger(T_{-})] = \sum_\alpha  p_{\alpha} \phi^{(\alpha)}_{T_{-}} ( Q_\alpha),
\eneq
where $\rho_{SB}(0)=\sum_\alpha p_\alpha Q_\alpha\otimes \rho_\alpha$ is any B+ decomposition of $\rho_{SB}(0)$ and
\beq
\label{indCPTPretro}
\phi_{T_-}^{(\alpha)}(\cdot):= {\rm Tr}_B[ U(T_-)(\cdot \otimes \rho_\alpha)U(T_-)^\dagger] ,
\eneq
and the backward time evolution operator $U(T_-):=\big(\mathcal{T} e^{-i \int_{T_-}^{0} H(s) ds}\big)^\dagger$ 
is the inverse of the forward time evolution operator that takes the system and bath from $T_-$ to $0$.
Note that $U(T_-)$ is unitary by definition, implying that $\phi_{T_{-}}^{(\alpha)}$ is a CPTP map for $T_-<0$, similarly to the case of forward time evolution. Thus Eqs.~(\ref{gendynamretro}) and~(\ref{indCPTPretro}) formally extend Eqs.~(\ref{gendynam}) and~(\ref{indCPTP}) to backward time evolution.  

Now, just as for the case of {forward evolution in Sec.~\ref{QSense} (see also~\cite{Dyson, Kubo, ME3, Spectro8}), the map $\phi_{T_{-}}^{(\alpha)}$ is a functional of the correlation functions $\langle B_{b_1}(t_1) \cdots B_{b_k}(t_k) \rangle_\alpha$, but  with $t_i \leq 0$. This can be seen, for example, by expanding Eq.~(\ref{indCPTPretro}) using the Dyson series, as is later done explicitly in Sec.~\ref{pertur}.} In principle, directly accessing these correlation functions requires measuring at times $t< 0$, which is forbidden in our scenario. However, this can be overcome if knowledge about $\langle B_{b_1}(t_1) \cdots B_{b_k}(t_k) \rangle_\alpha$  with $t_i \geq 0$ is available and certain mild conditions are satisfied. Let us see how in more detail. 

As described in Section~\ref{QSense}, knowlege about $\langle B_{b_1}(t_1) \cdots B_{b_k}(t_k) \rangle_\alpha$, for $t_i\in [0,T]$, \blk can be accesed via noise spectroscopy protocols that use control and measurements  in this interval. \blk Mathematically, these protocols lead to estimates, $\tilde{f}\blk^+_{\vec{b},\alpha} (\omega_1,\cdots, \omega_k)$, of the Fourier transforms $\langle B_{b_1}(\omega_1) \cdots  B_{b_k}(\omega_k)\rangle_\alpha$ of the correlation functions. Thus, the inverse Fourier transform \blk 
\beq \label{barecond}
{f}^+_{\vec{b},\alpha} (t_1,\cdots, t_k) \equiv \mathcal{F}^{-1}_{\vec{\omega},\vec{t}}\, (\tilde{f}^+_{\vec{b},\alpha} (\omega_1,\cdots, \omega_k))
\eneq reliably \blk estimates  $\langle B_{b_1}(t_1) \cdots B_{b_k}(t_k)\rangle_\alpha$ for times $0<t_i<T$. \blk

The crux of the matter is that knowledge of the correlators of $B_b(t)$ for times $t\in [0,T]$, \blk as captured by the estimate $f^+_{\vec{b},\alpha}(t_1,\cdots,t_k)$, does not generally guarantee knowledge of the correlators of $B_b(t)$ for $t<0$. \blk For example, $H_B(t)$ could be very \blk different for negative times, and we would never see the effect of this if we only have access to the system at times $t\geq 0$. {Thus we are interested in the the conditions under which it {\em is} possible to} extrapolate the values of the relevant correlators $\langle B_{b_1}(t_1) \cdots B_{b_k}(t_k)\rangle_\alpha$ to negative times, allowing us to compute the necessary quantities in Eqs.~\eqref{gendynamretro} and~\eqref{indCPTPretro} and thus successfully \blk retrodict the state of the system. 

Mathematically,  this is possible if the correlation functions \blk $\langle B_{b_1} (t_1) \cdots B_{b_k}(t_k)\rangle_\alpha$ are \blk  suffciently smooth functions of the $t_i$, so that the relevant correlator information acquired at positive times can be safely extrapolated \blk to negative times. This abstract condition imposes constraints on how $B_a(t)$ and its correlators can change in time and, as hinted earlier, is generally not satisfied. It is however a relatively mild assumption, requiring only that the relevant frequencies underlying the system evolution at past times fall within the range of frequencies accessed at later times. What is more, it holds in at least two natural settings, which we now describe. \blk

{\it Retrodiction for a quantum bath.---} Consider first a quantum bath whose evolution is ruled by a \blk constant bath Hamiltonian $H_B$ and which couples to the system via the same operators $\{B_b\}$ over the range $[T_0,T]$ (a very natural condition). In this case, one has that $B_b(t) = e^{i H_B t}B_b e^{-i H_B t}\,\,\, \forall \,\,t \,\,\in [T_0,T]$, which is typically  a smooth operator function of $t$. {It} is always smooth, with lower and upper frequency cutoffs, in the case of a finite-dimensional bath Hilbert space.  Say the estimators  $ \tilde{f}\blk^+_{\vec{b},\alpha} (\omega_1,\cdots, \omega_k)$ obtained via noise spectroscopy accurately sample all relevant frequencies {\blk (within the aforementioned cut-offs) of the correlators $\langle B_{b_1} (t_1) \cdots B_{b_k}(t_k)\rangle_\alpha$ during the interval $[T_0,T]$. Then it follows that ${f}^+_{\vec{b},\alpha} (t_1,\cdots, t_k)$ can be extrapolated to times $t\in[T_0,0]$ via Eq.~\eqref{barecond}, i.e., the smoothness of $B_b(t)$ for $t \in [T_0,T]$ guarantees that ${f}^+_{\vec{b},\alpha} (t_1,\cdots, t_k)$ is also a good estimate for $\langle B_{b_1} (t_1) \cdots B_{b_k}(t_k)\rangle_\alpha$ when $t_i \in [T_0,0]$. } \blk   It would then be possible to retrodict $\rho_S(t)$ for any $t \in [T_0,T]$.  One could also consider using alternative extrapolations of the correlators to negative times, e.g., via a Taylor series expansion instead of a Fourier expansion. \blk

{\it Retrodiction for a classical bath.---} In this case $B_b(t)$ corresponds to a classical stochastic process, i.e., a classical bath. { This differs from the above in the sense that it cannot be simulated by a constant bounded Hamiltonian.} Retrodiction is still possible, however, if one demands that the stochastic process is stationary.
\blk To show this, note first that any set of times $t_1,t_2,\dots,t_k\in[0,T]$ can always be trivially relabelled, via some permutation $P$, by $s_j:=t_{P(j)}$, such that $0\leq s_1\leq s_2\leq \dots\leq s_k\leq T$. Further, recalling that the noise is classical, one has $[B_{b}(t), B_{b'}(t')] =0\,\,\, \forall \,\, b,b',t,t'$, and hence  $\langle B_{b_1}(t_1) \cdots B_{b_k}(t_k) \rangle_\alpha=\langle B_{b_{P(1)}}(s_1) \cdots B_{b_{P(k)}}(t_k) \rangle_\alpha$. Moreover, since the noise is stationary, the correlators can only depend on relative time differences. Hence,  all future correlators, directly accessible via system control and measurement during $[0,T]$, are of the form
\beq \label{corr}
\langle B_{b_1}(t_1) \cdots B_{b_k}(t_k) \rangle_\alpha=
f_{\vec{b},\alpha}(\Delta_{1},\Delta_{2},\cdots,\Delta_{{k-1}}), 
\eneq	
with $\Delta_j:=s_{j+1}-s_j\geq 0$.	

Now, in contrast, for retrodiction in the interval $[-T,0]$ we require knowledge  of past correlators of the form $\langle B_{b'_1}(t'_1) \cdots B_{b'_k}(t'_k) \rangle_\alpha$ with $t'_j\leq 0$. These times can similarly be reordered, via some permutation $P'$, by $s'_j:=t'_{P'(j)}$, such that $s'_1\leq s'_2\leq\dots\leq s'_k\leq 0$. Stationarity then yields 
\beq \label{retrocorr}
\langle B_{b'_1}(t'_1) \cdots B_{b'_k}(t'_k) \rangle_\alpha = f_{\vec{b}',\alpha}(\Delta'_{1},\Delta'_{2},\cdots,\Delta'_{{k-1}}),
\eneq
with $\Delta'_j:=s'_{j+1}-s'_j\geq 0$. But the right hand side  may be recognised as being equal to a future correlator, as per  Eq.~(\ref{corr}), i.e., all past correlators for the interval $[-T,0]$ can be obtained from the directly accessible future correlators for the interval $[0,T]$. Thus,  we can succesfully retrodict as far as we can succesfully predict, and in particular whenever $|T_0|\leq T$. Note that for the quantum case, where $[B_{b}(t), B_{b'}(t')] \neq 0$ in general, the above argument does not hold.

These results can be generalised in several ways.  On one hand, one can bypass the assumption that the constant $H_S$ before and after $t=0$ has to be known, by considering a sufficiently powerful noise spectroscopy protocol. Noting that $H_{S} + H_{SB} = \sum g_{\alpha} \sigma_\alpha + \sum  \sigma_\alpha \otimes B_\alpha(t) \equiv \sum \sigma_\alpha \otimes \tilde B_\alpha(t)$, it is in principle possible to obtain all the necessary information, i.e., $H_S$ and the leading ${B}_\alpha(t)$ correlators,  from the $\tilde{B}_\alpha(t)$ correlators obtained from noise spectroscopy. On the other hand, when a {\it known} system-only control $H_{\rm ctrl}^{\rm (past)} (t)$ has been applied at $t<0$ our retrodictive power remains unchanged, since the correlators containing information about how the bath couples to the system remain unchanged. 

In summary, retrodiction of the system state $\rho(T_-)$, for $T_-$ in the range $[T_0,0] $, is possible whenever the relevant correlators $\langle B_{b_1} (t_1) \cdots B_{b_k}(t_k)\rangle_\alpha$ can be obtained for $t_i$ in this range. This is possible for sufficiently time-homogeneous quantum baths and for stationary classical noise processes.  It would be of interest to test this approach experimentally, via the tools of noise spectroscopy. For example, one could attempt to retrodict a {\it factorisable} initial system state prepared at time $T_0<0$ from measurements made during a time interval $[0,T]$. {Another avenue which could be explored is to extend our protocols to the prediction of the dynamics for $t>T$, when measurements are only available in $t \in [0,T]$. }\blk

\section{Application: Limited-access tomography}
\label{LAT}
\blk
We finish \blk our exposition by introducing a  one-sided or \blk {\it limited-access} tomography (LAT) protocol, in which information about a joint state is recovered by measuring only one of its subsystems, i.e., the ``probe'' subsystem, and by exploiting the information provided by its evolution under a known and partially controllable Hamiltonian, {\blk and the power of the $B+$ decomposition}. It is a quantum sensing protocol that falls in the generic category in Definition~\ref{defsp}, \blk i.e., when applied to factorisable initial states it \blk does not require any particular structure for the bath state in order to be successful. {\blk The expert reader will recognize that it is related to existing protocols (see for example Refs.~\cite{qLAT1, qLAT2}), but is strictly more powerful as it requires fewer assumptions to be executed while providing more information. We further note that, subsequent to the appearance of the current paper in the arXiv, Liu, Tian, Berttholz and Cai suggested a protocol (now published \cite{qLATPRL}) similar in spirit to ours, but that is  restricted to the initially factorizable scenario.} Thus, our LAT protocol is of strong interest in its own right and  \blk a good platform to demonstrate the value of \blk some of the tools we have \blk developed and discussed in the previous sections. 

We start by describing the general setting of our problem, namely a collection of qubits, evolving under a sufficiently rich Hamiltonian, {such that} we can only measure one of \blk them, i.e., the probe subsystem. We  show  that measuring the probe qubit at different times can yield information not only about the initial state of the probe, $\rho_S$, but also about the  joint \blk initial state, $\rho_{SB}$, where the remaining qubits  play the role of a \blk bath or environment relative  to the probe's evolution.  We describe how to do this by  first  introducing a LAT protocol  for initially uncorrelated states $\rho_{SB}=\rho_S\otimes\rho_B$ (to determine the unknown initial bath state $\rho_B$), \blk and then extend this protocol to arbitrary initial states $\rho_{SB}$ using the {methods} developed earlier in the paper.   We further show how control on the probe qubit can be used to ensure that even in potentially pathological situations, e.g., when the probe qubit couples very weakly to some of the a subsystems in the bath but not to the rest, the quality of our estimation of  $\rho_{SB}$ is not significantly affected. \blk  We conclude by presenting the results of numerically simulating the protocol in a physically relevant model. \blk

\subsection{The basic setup}

We consider a first \blk ``probe'' qubit interacting with a second \blk  $d$-dimensional bath (composed of one or many quantum systems), via a generic {\it known} Hamiltonian of the form 
\begin{equation}
\label{Hnat}
H_{\rm nat} = \sum_{\substack {a=0,x,y,z\\ b=0,\cdots,d^2-1}} g_{a,b} \,  \sigma_a \otimes W_b  =  H_S + H_{SB} +H_B.
\end{equation}
Here \blk $\{W_b\}$ is a basis for the linear operators acting on the \blk bath Hilbert space $\mathcal{H}_B$ (with \blk $W_0  =\bm 1_B$), and $H_{\rm nat}$ denotes the uncontrolled natural {Hamiltonian} of the probe and its bath. During a given experiment, we will allow the possibility of (i) adding fast control on the probe qubit, i.e., $H_{\textrm{ctrl},1} (t) = \sum h_a (t) \sigma_a$, and (ii) performing a tomographically complete set of measurements on the probe. In this way, the Hamiltonian in the interaction picture with respect to $H_{\textrm{ctrl},1} (t) + H_B$ can be rewritten as 
\begin{equation}
\label{Hamil}
\tilde{H}(t)= \sum_{a,b}  y_{a,b}(t) \sigma_b \otimes B_a(t),
\end{equation}
where $B_a = \sum_{b} g_{a,b} W_b,$ and 
\begin{align*}
B_a(t) &= U_{\textrm{frame}}(t)^\dagger B_a  U_{\textrm{frame}}(t) \\
y_{a,b(t)} &= \tr{U_{\textrm{frame}}(t)^\dagger \sigma_a  U_{\textrm{frame}}(t) \sigma_b}/2,
\end{align*}
with $U_{\textrm{frame}}(t) = \mathcal{T} [e^{-i \int_{0}^{t} ds (H_{\textrm{ctrl},1} ({s}) + H_B)}]$. The above implies that, in the Fourier domain, one can then write
$$
B_a(\omega) = \sum_{b,c} g_{a,b}  h_{b,c}(\omega, \vec{g}_{0,\cdot})  W_{c} .
$$
{where $\vec{g}_{0,\cdot} = \{ g_{0,1}, \cdots g_{0,d^2-1}\}$.} The function $h_{b,c}(\omega, \vec{g}_{0,\cdot})$ can be calculated directly from the above equations and, crucially, can be written as 
$$ h_{b,c}^{(r)}(\omega, \vec{g}_{0,\cdot}) =\sum_{s=\pm 1} C^{(b,c)}_{s} (\vec{g}_{0,\cdot})\, \delta (\omega +  s \Omega_r ),$$ 
where $\Omega_r$ is an effective resonance frequency that can be exactly or numerically calculated and $C^{(b,c)}_{s} (\vec{g}_{0,\cdot})$ is a computable coefficient. We say that the operator $B_a(\omega)$ is resonant at frequencies $\{ s \Omega _r \}$. In the case of a single-qubit bath, for example, the resonance frequency takes the form $\Omega_r\equiv 2\sqrt{ \sum_{l=x,y,z}  (g_{0,l})^2}$.   

\subsection{Factorisable initial states}

In order to facilitate the presentation, we start our discussion of the protocol by first considering the standard factorisable state scenario. We assume that at time $t=0$, where we start the analysis of our problem or we are given control, the system and bath are in a state of the form $\rho= \rho_S \otimes \rho_B$, with $\rho_B$ arbitrary, perhaps as a result of an appropriate preparation operation, e.g., a projective measurement on the system, performed at $t=0$. 

The objective of the LAT protocol in this scenario is to estimate the state $\rho_B$ \blk of the \blk  ``bath'' system,  by measuring only the ``probe'' qubit, in the presence of a known interaction (given by Eq.~\eqref{Hnat}). \blk Information about \blk the ``probe'' qubit state is assumed to be available via standard tomography at $t=0$. Following the notation of the previous sections, operationally the GSP is defined by the set of initial states $ \{\eta_\mu\}\blk =\{ \frac{\sigma_0 \pm \sigma_x}{2}, \frac{\sigma_0 \pm \sigma_y}{2},\frac{\sigma_0 \pm \sigma_z}{2}\}$, the observables $ \{O_\gamma\}\blk =\{\sigma_x,\sigma_y,\sigma_z\}$, and a set of control Hamiltonians $\{H_\beta\}$, which we will describe in more detail later.

To see how the protocol works and to understand the role control plays, it will be convenient to work in the interaction picture with respect to the control Hamiltonian and the purely bath Hamiltonian, as described earlier. The expectation values of interest can be written as 
\begin{align}
\label{Exact}
\nonumber E_{\eta_\mu,O_\gamma, H_\beta, \rho_B} &= \tr{U^{(\beta)} (\eta_\mu \otimes \rho_B) {U^{(\beta)}}^\dagger O_\gamma }\\
\nonumber &= {\rm Tr}_S [ {\rm Tr}_B [O_\gamma {U^{(\beta)}}^\dagger O_\gamma U^{(\beta)} \rho_B]  \eta_\mu O_\gamma]\\
&= \sum_{a,b}  V^{\gamma,\beta}_{a,b}   {\rm Tr}[  W_b \rho_B]  {\rm Tr} [ \sigma_a \eta_\mu O_\gamma]
\end{align}
where the coefficients $$V^{\gamma,\beta}_{a,b}= \tr{O_\gamma {U^{(\beta)}}^\dagger O_\gamma U^{(\beta)} (\sigma_a \otimes W_b)}/{2d}$$ and $ {\rm Tr} [ \sigma_a \eta_\mu O_\gamma]$ can be, in principle, exactly or numerically calculated for any choice of the control knobs $\eta_\mu,O_\gamma, H_\beta$. Notice that, for any such choice, the expression for $E_{\eta_\mu,O_\gamma, H_\beta, \rho_B}$ is simply an equation with known coefficients and unknown variables $\{ {\rm Tr}[ W_b \rho_B]\}.$ This is the working principle behind the protocol: by cycling over an appropriately chosen set of control knobs one can generate a solvable linear system of equations from which the ${\rm Tr}[ W_b \rho_B]$, with non-zero coefficient $K_b= \sum_{a}  V^{\gamma,\beta}_{a,b}  {\rm Tr} [ \sigma_a \eta_\mu O_\gamma]$ for at least one choice of $O_\gamma, \eta_\mu$ and $H_\beta$, can be extracted. 

For a given $H_{\rm nat}$ it may not be possible to find a set of control parameters such that all the $\tr{W_b \rho_B}$ are represented, e.g., a pathological Hamiltonian of the form $H_{\rm nat} = g \sigma_z \otimes \sigma_z + J \sigma_z$ . A sufficient condition to guarantee, {up to a change of basis}, that all the expectation values are represented in Eq.~\eqref{Exact} is that the Lie algebra of $\{B_b\}$, i.e., $\{B_{b_1}, [B_{b_1},B_{b_2}],[[B_{b_1},B_{b_2}],B_{b_3}], \cdots\}$, spans the whole operator basis for $\mathcal{H}_B$. Henceforth we assume we are dealing with a $H_{\rm nat}$ in this category. A physically relevant example of such a Hamiltonian for a composite system of $N$ qubits is \blk 
\beq 
\label{Hspec}
H_{\rm nat}^{(N)} \blk = \sum_{\substack {a=x,y,z\\ i,j=1,\cdots,N}} g^{(i,j)}_{a,a} \,  \sigma^{(i)}_a \otimes \sigma^{(j)}_a  + \sum_{ i=1,\cdots,N}  J_i \sigma^{(i)}_z,
\eneq
where the $g$'s and $J$'s are generic, i.e., they have no symmetries that would make $H_{\rm nat}$ unitary equivalent to a pathological Hamiltonian analog to the one described earlier. Here \blk $\{\sigma^{(i)}_b\}$ is the $b$-th Pauli matrix acting on the $i$-th qubit, and qubit 1 is identifed with \blk the probe.  This type of Hamiltonian, corresponding to a set of qubits coupled via a dipole-dipole interaction in the presence of a magnetic field, is ubiquitous in physics and, as noted in the Introduction, is particularly relevant to NV centers~\cite{NVcenter} and NMR~\cite{NMR}. \blk

\subsubsection{The role of control}

Control allows us to achieve the  desired \blk goal by: (i) ensuring that the system of equations includes all the variables of interest, (ii) providing a mechanism to generate sufficiently many equations to guarantee the solvability of the system of equations, and (iii) giving us the ability to build a well-conditioned system. 

In order to show how this is done, we turn to a perturbative analysis. While this path is not strictly necessary if we restrict ourselves to piecewise control (for which all unitaries can be exactly calculated), it is convenient to work in this language for generality, i.e., when an exact, closed-form expression for the expection values in the presence of time-dependent control cannot be obtained, and to facilitate some of the arguments. Moreover, we choose it because our long-term plan is to integrate the LAT protocol within a suite of system characterisation tools, which includes noise spectroscopy, that generally require perturbative expansions and the so-called filter function formalism~\cite{FFF, Spectro1, Spectro3, Spectro6, Spectro7}. 

\subsubsection{A perturbative expansion interlude} 
\label{pertur}
Using an adequate perturbative expansion (here we use the well-known Dyson series~\cite{Dyson}, but cumulant-like expansions are also possible~\cite{Kubo, ME3, Spectro8}), we can write the relevant expectation values (see Sec.\ref{QSense}) as \blk
\begin{align}
E_{\eta_\mu,O_\gamma, H_\beta, \rho_B}  &= {\rm Tr_S } \left[ \left(1+ \mathcal{D}^{(\beta)}_{1}(\rho_B,T) \right. \right.\nonumber\\
\label{Dys}
&\qquad\qquad + \left. \left. \mathcal{D}^{(\beta)}_{2}(\rho_B, T) +\cdots \right)  \eta_\mu O _\gamma \right],
\end{align}
where the Dyson terms are defined in the usual way~\cite{OpenQuantumSystemsBook2}, but with respect to the redefined Hamiltonian~\cite{Spectro8}
$$
H'(t) = \begin{cases}
-O_\gamma \tilde{H}(T-s) O_\gamma & {\rm for~~ } T\geq  t > 0\\
\,\,\,\,\,\,\,\,\,\,\,\,\tilde{H}(s+T) & {\rm for~~ } 0 \geq t \geq -T
\end{cases}~ 
$$
 with \blk $\tilde{H}(t)$ given by Eq.~\eqref{Hamil}. In this way one finds, for example, that   
\begin{align*}
\mathcal{D}^{(\beta)}_1 & = - i g  \int_0^T dt\,\, \langle \tilde{H}(t) - O_\gamma \tilde{H}(t) O_\gamma \rangle \\
\mathcal{D}^{(\beta)}_2 & = - g^2 \!\int_0^T \!\!\!dt_1 \int_{0}^{t_1} \!\!\!dt_2  \langle \blk O_\gamma \tilde{H}(t_2) \tilde{H}(t_1) O_\gamma +\tilde{H}(t_1) \tilde{H}(t_2) \\
&\,\,\,\,\,\,  - O_\gamma \tilde{H}(t_1) O_\gamma \tilde{H}(t_2)- O_\gamma \tilde{H}(t_2) O_\gamma \tilde{H}(t_1)    \rangle . \blk 
\end{align*}

Moving to the frequency domain, the Dyson terms can be rewritten in terms of the purely control dependent filter functions 
$$
F_{\vec{a},\vec{b}}^{(k)} (\omega_1, \cdots, \omega_k,t) = \int_{0}^{t} \!\! d_> \vec{s}_{[k]} \,\prod_{j=1}^k y_{a_j,b_j}(s_j) e^{i s_j \omega_j},
$$
where we have used $\int_{0}^{t} d_>\vec{s}_{[k]}$ to denote the ordered integration, i.e., $s_1 \geq s_2 \geq \cdots \geq s_k$, and the Fourier transforms of the moments
$$
M^{(k,r)}_{\vec{a}}(\vec{\omega}) \equiv  \langle B^{(r)}_{a_1}(\omega_1) \cdots B^{(r)}_{a_k}(\omega_k)\rangle,
$$
such that, for example, 
$$
\mathcal{D}^{(\beta)}_1  = \frac{-i g}{2 \pi}   \int_{-\infty}^\infty \!\!\!\!\!\! d\omega   \sum_a   F_{a,b}^{(1)}(\omega,T) (\sigma_a - O_\gamma \sigma_a O_\gamma) \, M^{(1,r)}_a(\omega). 
$$	
More generally, each Dyson term will be typically written as a linear combination of convolutions of the form 
\beq \label{ik}
 I_k = \int d\vec{\omega} F^{(k)}_{\vec{a},\vec{b}} (\vec{\omega},T) M^{(k,r)}_{\vec{a}}(\vec{\omega}). 
 \eneq 
 This is typical of noise spectroscopy protocols and, indeed, finding ways to reliably deconvolute such integrals is one of the main roadblock when designing them for general baths. This is where the knowledge of the structure of the bath of our problem, i.e., its finite dimensional character, becomes important. The key thing to observe is that when working in the Fourier domain, each moment can be written as 
\begin{align} 
\nonumber M^{(k,r)}_{\vec{a}}(\vec{\omega}) &= \sum_u \sum_{\vec{a},\vec{b},\vec{c}} \Big(  \prod_j  g_{a_j,b_j} C^{(s_j)}_{b_j,c_j} (\vec{g}) \delta(\omega_j + s_j \Omega_r)\Big) \\
& \,\,\,\,\,\,\,\,\, \times  \frac{\tr{ \sigma_{c_1} \cdots \sigma_{c_k} \sigma_u}}{2} \langle W_u \rangle.
\end{align}

\subsubsection{Control as a tool to generate a well-conditioned system}

Having briefly introduced the relevant details about the perturbative expansion, we are now ready to show that, aided by control, we can extract  the desired information from measurements on the probe qubit only. 

The first role that control plays is to allow us to generate multiple linearly independent equations, which ultimately allow us to build a solvable linear system via Eq.~\eqref{Exact}. From the perturbative point of view, fast control on the probe means that the filter functions can be changed or, equivalently, that the coefficients $V_{a,b}^{\gamma,\beta}$ in Eq.~\eqref{Exact} can be further manipulated while keeping $O_\gamma$ and $\eta_\nu$ constant. This is enough then to be able to, at least in principle, extract the desired $\{\tr{W_b \rho_B}\}$. In practice, however, where operations are imperfect it may not be enough: one has to also ensure that the system of equations is as well conditioned as possible.

To illustrate why a well-conditioned system is important, i.e., that the inferred values of $\langle W_b \rangle$  are robust to fluctuations in $E_{\rho,O_\gamma, H_\beta, \rho_B}$, consider a simple Hamiltonian of the form 
$$
\tilde{H}(t) =  \sum_{j=1,2} y_{3,j}(t)  g_j \sigma_z \otimes  B_j(t), 
$$
with $B_1(t) = \sigma_z$ and $B_2(t) =  \cos[\Omega t ] \, \sigma_x + \sin[\Omega t ] \,  \sigma_y$, i.e., $B_1(\omega)$ is resonant at $\omega=0$ while $B_2(\omega)$ is resonant with $\omega = \pm \Omega$, and $y_{3,j} (t)$ is the switching function resulting from applying a sequence of $\sigma_x$ pulses. The problem stems from the fact that if, for example $g_1  \gg \blk g_2$, under free evolution the effect on the probe of the term proportional to $g_{2}$ is negligible compared to the other one. In broad terms, this implies that the values of $\tr{\sigma_x \rho_B}$ or $\tr{\sigma_y \rho_B}$ can only be accurately recovered if the fluctuations in $E_{\eta_\mu,O_\gamma, H_\beta, \rho_B}$, due to experimental errors for example, are much smaller than the typical value of terms involving $\tr{\sigma_x \rho_B}$. This is not ideal as it may lead to bad estimates. Control can be used to combat this problem. 

The idea is to use a control sequence that is ``resonant'' with a bath operator in order to enhance its contribution. In order to see how this work it is convenient to gain insight from the perturbative approach based around the filter function formalism. Since $g_{1}  \gg \blk g_{2}$ and $F^{(1)}_{3,1} (\omega,T) =F^{(1)}_{3,2} (\omega,T)$, one expects that $\| g_{1} \int d\omega  F^{(1)}_{3,1} (\omega,T)  \tr{B_1(\omega) \rho_B}\|   \gg \blk  \| g_{2} \int d\omega  F^{(1)}_{3,2} (\omega,T)  \tr{B_2(\omega) \rho_B}\|$, i.e.,  that \blk the term proportional to $g_1$ is dominant. That is, unless the filters have a specific structure, the contribution of $\tr{ \sigma_x \rho_B}$ and $\tr{ \sigma_y \rho_B}$ to the probe dynamics will be overshadowed by the $\tr{ \sigma_z \rho_B}$ contribution. However, by noticing that $B_1(\omega) \propto \delta (\omega)$ and $B_2(\omega) \propto \delta (\omega \pm \Omega)$,i.e., the different bath operator have different resonance frequencies, we can drastically modify this situation. Indeed, if one chooses a decoupling sequence of cancellation order $\delta \neq 0$, then one has that $F^{(1)}_{3,1} (\omega=0,T) = 0$ and that, in general,  $F^{(1)}_{3,1} (\omega= \Omega,T) \neq 0$, i.e., the generally dominant contribution proportional to $g_{1}$ can be suppressed. But one can take this a step further, one can simultaneously ensure that the contribution of the term proportional to $g_{2}$ is large. By using a decoupling sequence composed of basic sequence of length $T_c$ repeated $M  \gg 1$ \blk  times, one gets that  
\begin{align}
\label{pref}& \Vert g_{2} \int d\omega  F^{(1)}_{2} (\omega, M T_c)  \tr{B_2(\omega) \rho_B}\Vert \\
\nonumber &= \Vert g_{2} \int d\omega  \frac{ 1 - e^{i M \omega T_c}}{1 - e^{i \omega T_c}} F^{(1)}_{3,2} (\omega,T_c)  \tr{B_2(\omega) \rho_B}\Vert.
\end{align}
One can verify that the $\frac{ 1 - e^{i M \omega T_c}}{1 - e^{i \omega T_c}}$ factor plays the role of a window function (both its real and imaginary part) that grows with $M$ around $\omega = k 2 \pi/T_c$, for $k=0,1,\cdots$ while suppressing other frequencies. The width of this window decreases with $M$ while its height grows with $M$. Thus, by repeating a dynamical decoupling sequence with $T_c = 2 \pi / \Omega$ we can enhance the contributions of bath operators which are resonant with the control sequence, i.e.,  $\omega = k 2 \pi/T_c = k \Omega$, while suppressing the rest. In other words, for an appropriate choice of control one can then make the contribution proportional $g_{2}$ term be the dominant one. Thus, by adding the equations to the ones using free evolution, for example, one can avoid the ill-conditioned system that would result from the described Hamiltonian.

The situation can be more complicated for a general Hamiltonian, e.g., one may need to hit multiple resonances, but the same principle applies: by using a control sequence that matches the resonance frequency of the different bath operators (induced by the bath only part of the Hamiltonian) can suppress or enhance the relative effect of different Hamiltonian terms on the probe dynamics.

\subsection{Extension to arbitrary initial states} 

If the initial state of system plus bath $\rho_{SB}$ is   allowed to be \blk correlated, then we can use the recipe  for extending a QSP \blk  in Theorem~\ref{thm2} of \blk Sec.\ref{QSense}, to generalise the above limited-access tomography protocol  to determine $\rho_{SB}$, via local operations and measurements on the system only. \blk 

As before, we first notice that, by using a B+ decomposition for the initially unknown state $\rho_{SB}= \sum w_\alpha Q_\alpha \otimes \rho_\alpha$, experimentally we have access to $\sum_\alpha w_\alpha E_{Q_\alpha ,\sigma_z, H_\beta, \eta_\alpha}$. As discussed  in the proof of Theorem~\ref{thm2}, \blk by using our freedom to apply an operation at time $t=0$ we can generate a system of equations from which the individual $w_\alpha$ and $E_{Q_\alpha ,\sigma_z, H_\beta, \rho_\alpha}$ can be obtained.   An example of   a suitable \blk set of local operations, when the probe susbsytem corresponds to a single qubit, is the following~\cite{modiexp}. Consider the  probe  states $\{ \ket{{s_a,\sigma_a}} \}$, for $s_a = \pm$ and $a= x,y,z$, where $\ket{{\pm,\sigma_a}}$ is the $\pm$ eigenstate of the $\sigma_a$ operator. Then, the set of CPTP maps $\{\mathcal{R}_{ s_a, \sigma_a; s_b, \sigma_b}\}$, with $\mathcal{R}_{ s_a, \sigma_a; s_b, \sigma_b}$ corresponding to \blk a projection  onto \blk $\ket{{s_a,\sigma_a}} \bra{{s_a,\sigma_a}}$ followed by the unitary rotation that takes $\ket{{s_a,\sigma_a}}$ to $\ket{{s_b,\sigma_b}}$, is enough to allow the recovery of the $E_{Q_{\alpha'} ,\sigma_z, H_\beta, \rho_\alpha}$.  \blk

Because the $Q_\alpha$  in the B+ decomposition \blk form a basis  set, \blk one can then calculate 
$$
E_{\eta_\mu,\sigma_z, H_\beta, \rho_\alpha} = \sum_{\alpha'} \tr{\eta_\mu P_{\alpha'}} E_{Q_{\alpha'} ,\sigma_z, H_\beta, \rho_\alpha},
$$
for any $\eta_\mu$. From this, as in the initially factorisable case, we can obtain the value for the quantities analogue to the $\{ I_k\}$  in Eq.~\ref{ik}, \blk but with ${\rho_B \rightarrow \rho_\alpha}$, i.e.,  with the \blk moments $\tr{B_{a_1}(\omega_1) \cdots B_{a_k}(\omega_k)\rho_B }$ replaced by $\tr{B_{a_1}(\omega_1) \cdots B_{a_k}(\omega_k)\rho_\alpha }$. In turn, this implies that by using our ability to control the probe system we can access the $\tr{\sigma_a \rho_\alpha}$, for all $a$ and $\alpha$. Finally, in order to reconstruct the density matrix at $t=0$, we note that 
\begin{align*}
{\rm Tr} [(\sigma_a \otimes \sigma_{a'}) \rho_{SB}(t=0)]  &= \sum_{\alpha } w_\alpha {\rm Tr} [\sigma_a Q_\alpha ] {\rm Tr} [\sigma_{a'} \rho_\alpha].
\end{align*}
Since the ${\rm Tr} [\sigma_a Q_\alpha ]$ can be calculated, and the described protocol gives us access to the $w_\alpha$ and the $\tr{\sigma_{a'} \rho_\alpha}$, we have all the information necessary to do tomography in the initially correlated joint state, as claimed.

\subsection{Illustrative example}

In order to illustrate the protocol described in this section and deploy all the tools discussed, we consider the scenario of a probe qubit ($i=1$) coupled to two bath qubits ($i=2,3$) via a Hamiltonian of the form of Eq.~(\ref{Hspec}): \blk
\beq
\label{Hspec1p2}
H_{\rm nat}^{(3)} \blk = \sum_{\substack {a=x,y,z\\ i \neq j=1,2,3}} g^{(i,j)}_{a,a} \,  \sigma^{(i)}_a \otimes \sigma^{(j)}_a  + \sum_{i=1,2,3}  J_i \sigma^{(i)}_z,
\eneq
working in units where $J_1 = 0, J_2= 1, J_3=3$, $g^{(1,2)}_{a,a} = 1= g^{(2,3)}_{a,a}$, and $g^{(1,3)}_{a,a}=1/100$, i.e., the probe qubit couples much more strongly to one of the ``bath'' qubits than to the other. { We consider, as an example, the situation where at $t=0$ the three qubit system is initialised in state $\ket{\Psi}= \frac{\ket{001} + \ket{010} +\ket{100}}{\sqrt{3}}$ and the preparation procedure consists of a $\sigma_z$ measurement made on the probe qubit, yielding the state $\rho_{SB} = \frac{2}{3} \ket{0}\bra{0} \otimes \ket{\phi_{+}}\bra{\phi_{+}} + \frac{1}{3} \ket{1}\bra{1} \otimes \ket{00} \bra{00},$ with $\ket{\phi_{+}} = \frac{ \ket{01} + \ket{10}}{\sqrt{2}}$.} The state is then allowed to evolve until $t=T_1$, at which point one is given control over the probe qubit and can execute a LAT protocol. 

The ingredients of the protocol are as follows. As local maps at time $t=T_1$, we choose the $\{\mathcal{R}_{ s_a, \sigma_a; s_b, \sigma_b}\}$
described in the previous subsection. As measurements on the probe qubit we choose the complete set of Pauli operators $\{\sigma_a\}$, for $a=x,y,z$. Finally, as fast control on the probe qubit, we choose concatenated dynamical decoupling (DD) sequences of cancellation order $\delta=2$, i.e., CPMG sequences~\cite{CPMG1,CPMG2}, of cycle time $T_c$ repeated $M$ times. The cycle times are chosen to enhance/suppress the different bath operators and thus we build sequences that match the resonance frequencies given by 
\begin{align*}
\{\Omega_r\} &= \{\Big( \sqrt{(g_{x,x} + g_{y,y} + g_{z,z})^2 + (J_2 - J_3)^2}\\
& \,\,\,\,\,\,\,\,\, \pm \sqrt{(g_{x,x} - g_{y,y} + g_{z,z})^2 + (J_2 + J_3)^2} \Big)\\
& \simeq \{7.30318, 1.63217\}.
\end{align*}
To further suppress the contribution of different terms in the Hamiltonian in a given experiment, we also cycle over sequences that use different types of pulses, i.e., $\sigma_x, \sigma_y,\sigma_z$. { Note that a sequence composed of $\sigma_\alpha$ pulses suppresses the contribution of terms that anticummmute with it, so by cycling over the different types of sequences we alternatively supress the contribution of some terms in the Hamiltonian relative to others.} Once we have generated the necessary equations, we implement a naive routine to obtain the estimates of $\tr{W_a \rho_B}$. We first use an unconstrained Least Squares algorithm in order to solve the system of equations and obtain estimates $\{ \tr{W_a \rho_\alpha}\vert_{\rm estimate}\}$. To obtain a physical system-bath state, we initially use these estimates to build the operator $\tilde{\rho}_{SB}$, which may not be positive, and finally pick the closest positive operator~\cite{CloseP} $\bar{\rho}_{SB}$ as our true estimate~\footnote{We recognise that the data processing inequality implies this routine is not optimal and other processing options are possible, e.g., as in Ref.~\cite{Spectro12}, but it is not our objective here to optimise the protocol.}. 

For simplicity in the analysis we only introduce additive errors in the final measurements and we do so by each time picking a small correction from a Gaussian distribution of mean $\mu=0$ and variance $\sigma^2= 1/10$ and averaging over $100$ realisations. Under these assumptions, we find that estimation error of the three-qubit state $\rho_{SB}$ induced by the measurement error can be of the same order of the one we would find if we were to to do single qubit tomography on the probe, provided the control is adequately chosen. {In our example, the quality of the control is given by how narrow the window induced by the $ \frac{ 1 - e^{i M \omega T_c}}{1 - e^{i \omega T_c}} $ prefactor in Eq.~\eqref{pref} is, i.e., if the different resonance frequencies were closer to each other one would need to choose a narrower window.}

For example, in a sample run of the protocol, using only sequences composed of $\sigma_x$ pulses repeated $M=10$ times but using only one resonant frequency, we find that the fidelity between the estimated state and the actual $\rho_{SB}(T_1)$ to be $\mathcal{F}_{SB}(T_1)=0.8207$, while the fidelity of the estimated $\rho_S(T_1)$ and the actual state would be $\mathcal{F}_{S}(T_1)=0.9991$ if we were to do just tomography of the probe using the same sort of noisy single qubit measurements. {Notice that it makes sense to compare these two scenarios as they both use single qubit measurements on the probe qubit. It would seem from the above numbers that there is a price to pay for infering the three qubit state from measurements on the probe qubit only. However, as discussed earlier, control can overcome this problem. }

In contrast, when we use both resonant frequencies (still in the $M=10$ regime), we find that $\mathcal{F}_{SB}(T_1)=0.9865$, in agreement with the expected better conditioning of the system induced by control. If we further use the system generated by a full set of decoupling sequences cycling over all axis, using both resonant frequencies, and $M=50$ repetitions, one finds that $\mathcal{F}_{SB}(T_1)=0.9952$. That is, the better designed the control is the better conditioned the resulting system of equations is. It should be pointed out that the type of control, i.e., repetition of a base sequence, we have proposed here is perhaps the simplest way of addressing the issue and is, by no means, unique. Other options, such as using a different set of base sequences and a fixed, but large, number of repetitions as in Refs.~\cite{Spectro6,Spectro7} or even more advanced forms of filter design as in Refs.~\cite{Spectro10,Spectro11} may yield better results, but ultimately the ideal choice will be determined by the control capabilities available to the experimenter. Here, our interest was to show LAT protocols are possible. 

The information obtained from the LAT protocol allows us also to illustrate our point  in Sec.~\ref{retrosec} \blk on the retrodiction of the probe dynamics. In the simple scenario in which the Hamiltonian is fully known, the information about $\rho_{SB} (T_1)$ not only allows us to reconstruct $\rho_S(0)$ but the whole  initial system-bath state \blk $\rho_{SB} (0)$. Notice that, since the fidelity is invariant under unitary operations,  $\mathcal{F}_{SB}(0) =\mathcal{F}_{SB}(T_1)$. It should be pointed out, however, that  this property, of \blk the quality of the estimate of $\rho_S(0)$ being independent of the value $T_1$, is not a feature to be expected in more general quantum sensing protocols, that yield information about only the leading bath correlators and depend on the convergence  of the \blk perturbation expansion. 

\section{Conclusion}

In this paper, we have \blk studied the dynamics of an open quantum system in the general scenario where $\rho_{SB}(0) \neq \rho_S \otimes \rho_B$. By introducing a universal decomposition for an arbitrary state $\rho_{SB}$, we showed that techniques previously developed and well studied within the factorisable initial state context---e.g., mathematical objects (such as CPTP maps), calculational methods to approximate the dynamics of an open quantum system (such as master equations),  definitions of Markovian evolution, \blk and protocols to sense an environment by measuring the response of an open quantum system (such as noise spectroscopy protocols)---can all be seamlessly extended to the general scenario in which initial system-bath correlations are present. \blk Moreover, we fully solved the qubit dephasing model as an indicative example. \blk 

We have  further \blk applied our methods to present \blk a new `Limited Access Tomography' protocol which provides a way of performing tomography on a multipartite \blk system, evolving under a sufficiently rich and known Hamiltonian, via \blk measuring and controlling only a ``probe'' subsystem. Another interesting application of our results is the possibility of retrodicting the dynamics of a quantum system undergoing stationary noise, from measurements performed in the future. That is, we show that if a experimenter receives a state at time  $t=0$, \blk which is potentially entangled with its environment, then it is possible reconstruct the density operator of the system at  an earlier time, \blk e.g., the output system state of a preparation procedure at $t=T_0<0$, by exploiting information gathered from measurements performed at times  $t>0$. 
 
We expect that the results presented here will provide a direct way to extend existing open quantum system tools to the non-factorisable initial system and bath state scenario, and facilitate the generalisation of results in that context. For example, it should be possible to analyse quantum thermondynamics and heat transport problems in the scenario were the system and the reservoir(s) start in a correlated state. Similarly, from the point of view of quantum sensing, we expect that our results will open the way for novel protocols in which the degrees of freedom of interest are encoded in initially correlated system-bath states. For example, consider a qubit lattice of which one can only access a set of sites, which \blk is in the ground state of its Hamiltonian. With the results presented here, it would be possible then to `sense' properties that are encoded in the correlations between the subsystems, e.g., to \blk characterise area-laws~\cite{EisertRMP}. 

In the arena of noise spectroscopy, \blk the tools developed here will allow the characterisation of \blk bath correlations irrespective of initial conditions. These tools also allow significant extensions of \blk other noise characterisation protocols. For example, they \blk in principle allow randomised benchmarking  to be extended beyond the assumption that all noisy unitaries are represented by CPTP maps. Within the context of randomised benchmarking, this assumption has hitherto imposed a very strong constraint on the possible correlations that can be generated between the system and bath during their joint evolution, and thus on the noise mechanism capable of generating such correlations. 

This new capability will become more and more important as theoretical and experimental efforts push towards the higher quality system characterisation~\cite{chara1} needed to achieve significantly-below-threshold fidelities.  

\section{Acknowledgements}

This project received grant funding from the Australian Government via the AUSMURI grant AUSMURI000002. GPS is pleased to acknowledge support from the DECRA fellowship DE170100088. \blk The ARC Centre of Excellence Grant No. CE170100012 (CQC2T) supported the research contributions of GPS and HMW. 

 We thank our anonymous referees for their valuable comments, which greatly improved the exposition of our results, and for pointing our attention to related literature (e.g. Refs. ~\cite{qLAT1, qLAT2, deco1, deco2}). \blk



\appendix
\section{Constructing B+ decompositions}
\label{construct}

The construction of general and canonical B+ decompositions in Sec.~\ref{theory} relies on finding a dual frame or dual basis $\{Q_\alpha\}$ of system operators, for a given frame or basis $\{ P_\alpha\}$ of positive operators.  The general theory of frames on vector spaces is well established~\cite{frames}, and has found applications in quantum tomography (see for example Ref.~\cite{optTomogra} for a comprehensive review). Indeed,  if one has POVMs $\{E^{(1)}_r\},\{E^{(2)}_s\},\dots$, which form a tomographically complete set, then one can take $\{ P_\alpha\}$ to be the concatenated set $\{\dots,E^{(1)}_r,\dots, E^{(2)}_s,\dots\}$. Here we give the details needed for general B+ decompositions, with canonical B+ decompositions discussed in Appendix~\ref{condual}.

As a convenient reference set, let $\{G_j\}$ be some orthonormal basis set of Hermitian operators acting on the Hilbert space $\mathcal{H}_S$ of the system, so that
\beq \label{ortho}
\tr{G_j G_k}=\delta_{jk},~~~~~G_j^\dagger=G_j. 
\eneq
It is convenient to represent this basis by a vector operator, $\bm G$, with the $j$th component given by $G_j$. For the case of a $d$-dimensional Hilbert space $\bm G$ has $d^2$ components.  It follows from Eq.~(\ref{ortho} that any operator $A$ can be written as
\beq \label{ortho2}
A=\sum_j \tr{AG_j}G_j = \tr{A\bm G^\top}\bm G .
\eneq

An arbitrary (and possibly overcomplete) basis set or frame $\{P_\alpha\}$ is similarly be represented by a vector operator $\bm P$, which for the case of a $d$-dimensional Hilbert space has at least $d^2$ components. It can be uniquely expressed in terms of the basis set $\{G_j\}$,  using Eq.~\ref{ortho}), as
\beq \label{pag}
\bm P = \T\,\bm G ,\qquad \T:= \tr{\bm P \bm G^\top},
\eneq
Note that  $\T$ is a real matrix, and will be non-square when $\{P_\alpha\}$ is overcomplete.

To construct a dual frame $\{ Q_\alpha\}$ satisfying Eq.~(\ref{decom}) (i.e., $A = \sum_{\alpha} \tr{ A Q_\alpha} P_\alpha =  \sum_{\alpha} \tr{ A P_\alpha} Q_\alpha$), represented by vector operator $\bm Q$, note first that we similarly must have
\beq \label{qtg}
\bm Q =\tilde \T\,\bm G ,\qquad \tilde\T:=  \tr{\bm Q\bm G^\top}
\eneq
for some real matrix $\tilde\T$ having the same dimensions as $\T$. Substituting this into Eq.~(\ref{decom}), and using Eq.~(\ref{ortho2}), then gives
\[
\tr{A\bm G^\top}\bm G = \tr{A\bm P^\top}\bm Q = \tr{A\bm G^\top} (\T^\top \tilde\T) \bm G.
\]
Since $A$ is arbitrary, and has a unique expansion with respect to $\bm G$, it follows that
$\T^\top \tilde\T$ is equal to the $d^2\times d^2$ identity matrix
(with $d\equiv\infty$ for an infinite-dimensional Hilbert space).  This is easily checked to have the solution
\beq \label{tildet}
\tilde\T = \T\, (\T^\top \T)^{-1} 
\eneq
when the inverse exists (which it typically does, as discussed below). Equation~(\ref{qmp}) of the main text then follows via
\[ \textrm{M}\bf P = \T(\T^\top \T)^{-2}\T^\top\bm P=\T(\T^\top \T)^{-1}\bm G=\tilde\T\bm G = \bm Q. \]
For an overcomplete basis $\{P_\alpha\}$ there is in fact an infinite set of solutions for $\textrm{M}$, each with a corresponding dual frame~\cite{frames}, reflecting the fact that the decomposition of a state in such a basis is not unique. 

Finally, it may be shown the inverse of $\T^\top \T$ always exists for a finite Hilbert space, and also exists for an infinite-dimensional Hilbert space under a mild `frame' condition on $\{P_\alpha\}$.
In particular, noting 
\[ (\T^\top \T)_{jk}= \sum_\alpha \tr{P_\alpha G_j}\,\tr{P_\alpha G_k} = \bm v^{(j)}\cdot \bm v^{(k)}, \]
with $\bm v^{(j)}:=\tr{\bm P\,G_j}$, it follows that $\T^\top \T$ is a Gram matrix with respect to the vectors $\{ \bm v^{(j)}\}$. Hence, this matrix is strictly positive if and only if the $\bm v^{(j)}$ are linearly independent. But linear independence is guaranteed, since $\sum_j c_j\bm v^{(j)}=\tr{\bm P\sum_j c_jG_j}$, and so can vanish only if $c_j=0$ for all $j$ (since the $ P_\alpha$ form a basis and the $G_j$ are linearly independent). Hence, 
\beq \label{pos}
\T^\top \T>0 .
\eneq
It follows immediately that the inverse in Eq.~\ref{tildet} exists for a $d$-dimensional Hilbert space with finite $d$. Further, for an infinite Hilbert space, the inverse exists if the eigenvalues of the Gram matrix are bounded above and below, i.e, if there are finite constants $a,b>0$ such that
\beq a\leq \T^\top \T \leq b. 
\eneq
Whether this condition is satisfied depends on the particular frame $\{P_\alpha\}$, as can be seen explicitly by multiplying on the right by the vector $\tr{A\bm G}$ and on the left by the transpose $\tr{A\bm G^\top}$, and using Eqs.~(\ref{ortho})--(\ref{pag}), to rewrite it as the equivalent `frame condition'~\cite{frames}
\beq a \leq \sum_\alpha \tr{AP_\alpha}^2 \leq b  \eneq
for all Hermitian operators $A$ satisfying $\tr{A^2}=1$.

\section{Constructing canonical B+ decompositions}
\label{condual}

A canonical B+ decomposition requires the $P_\alpha$ to (i)~be linearly independent and (ii)~to sum to the unit operator. Condition~(i) implies the mapping~(\ref{pag}) between the two basis sets\textbf{} is one:one, from which it immediately follows that $\T$ is invertible and $\bm G=T^{-1}\bm P$. It then follows from Eqs.~(\ref{pag}) and~(\ref{qtg}) that there is a unique dual basis, given by
\beq \label{qgen}
\tilde T = (\T^\top)^{-1} ,\qquad \bm Q=\tilde \T\bm G=\tilde \T\T^{-1}\bm P= (\T\T^\top)^{-1}\bm P,
\eneq
as noted in Sec.~\ref{canonbplus}. We remark that the biorthogonality relation~(\ref{biorthog}) can alternatively be written in matrix form, as $\tr{\bm P\bm Q^\top}=I_{d^2}$. It is also worth noting that while the dual basis elements cannot all be positive in general~\cite{Ferrie08}, the second condition implies that they all must have unit trace:
\beq \tr{Q_\beta}=\sum_\alpha \tr{P_\alpha Q_\beta}=\sum_\beta \delta_{\alpha\beta}=1 .
\eneq

Note that condition~(ii), which requires $\{ P_\alpha\}$ to be a POVM, is a relatively minor restriction, as a suitable POVM can be constructed from any linearly independent basis set of positive operators $\{ P'_\alpha\}$ (which may in turn be constructed from a suitable subset of a tomographically complete set of operators).  In particular, for any such linearly independent basis set, define $P':=\sum_\alpha P'_\alpha$.  This operator is not only positive but is strictly positive, i.e., $P'>0$. To see this, suppose there is some state $|\psi\rangle$ such that $P'|\psi\rangle=0$.  Hence $\langle\psi|P'|\psi\rangle= \sum_\alpha \langle\psi|P'_\alpha|\psi\rangle=0$. But each term in the sum is non-negative, implying $\langle\psi|P'_\alpha|\psi\rangle=0$.  Equation~(\ref{decom}) then gives
$$|\psi\rangle\langle\psi|=\sum_\alpha Q'_\alpha \tr{P'_\alpha |\psi\rangle\langle\psi|}=\sum_\alpha Q'_\alpha \langle\psi|P'_\alpha |\psi\rangle  = 0,$$
and thus $|\psi\rangle=0$, i.e., $P>0$ as claimed. Finally, defining 
\beq
P_\alpha := (P')^{-1/2} P'_\alpha (P')^{-1/2}
\eneq
gives an informationally complete POVM $\{ P_\alpha\}$ as required.

Equation~(\ref{qgen}) for the dual basis may be bypassed for the case of symmetric informationally complete POVMs (SIC POVMs)~\cite{sicpovm}, for which the high degree of symmetry allows the $\{Q_\alpha\}$ to be evaluated explicitly. In particular, a POVM $\{P_\alpha\}$ with $d^2$ linearly independent operators is defined to be a general SIC POVM if and only if $\tr{P_\alpha^2}=$ constant and $\tr{P_{\alpha}P_{\beta}}=$ constant, for all $\alpha\neq\beta$~\cite{Gour14}. Defining $a=\tr{P_\alpha^2}$, the corresponding dual basis is then given by
\beq
Q_\alpha = \frac{d}{ad^3-1} \left[ (d^2-1) P_\alpha - (1-ad){\bm 1}_S \right],
\eneq
as may be verified by direct calculation~\cite{Gour14}. This generalises Example~\ref{Ex2} of the main text, which corresponds to the case $d=2$ and $a=1/4$.

\section{Searching for computational Markovianity}
\label {DMark}
We discuss here how one would search for a frame in which the evolution is computationally Markovian, as per Definition~\ref{defdm} of Sec.\ref{sec:dm}. We distinguish between two types of definitions $\mathcal{M}[H(t),\rho_B]$ of Markovianity for the case of initially uncorrelated states: those which require explicit knowledge of $H(t)$ and $\rho_B$, and those that only require knowledge of the dynamical map $\phi_t (\cdot)$ induced by the natural dynamics and, possibly, by appropriate interventions. For the first type of definition, verifying computational Markovianity in the non-factorisable case is then a matter of finding an appropriate frame, with $\rho_{SB} = \sum w_\alpha Q_\alpha \otimes \rho_\alpha$, such that $\{H(t), \rho_\alpha\}$ gives rise to a corresponding Markovian evolution,  i.e., such that $\mathcal{M}[H(t),\rho_\alpha]$ is satisfied if $w_\alpha>0$. The second class, on which the following discusion will concentrate, is perhaps more interesting as it includes the definitions that can be in principle experimentally verified. We now describe how one would proceed in this case. 

Let us first note that, in the factorisable case, verifying that $H(t)$ and $\rho_B$ satisfy some Markovian definition or criterion $\mathcal{M}[H(t),\rho_B]$ requires \blk (i) evolving $\rho_S\otimes \rho_B$ with  $H(t)$   (plus possibly an appropriate set of interventions on the system (or even on the bath if experimentally accessible), for various times $t$ and system states $\rho_S$, and (ii) using some sort of classical processing to verify that the set of dynamical maps associated to such evolutions, say $\phi_{t;i}$, satisfy a mathematical condition specified by $\mathcal{M}[H(t),\rho_B]$. For example, for the case of Markovianity defined via the divisibility of the evolution between two times $t>t'>0$~\cite{RHP2010}, one has to (i) generate $\phi_t$ and $\phi_{t'}$ in Eq.~(\ref{CPTPfact})  (e.g., via evolving $\rho_S\otimes\rho_B$ for some basis set of system states $\{\rho_S\}$), and (ii) verify that  they satisfy the condition $\phi_t = \Phi_{t,t'} \circ \phi_{t'}$ for some CPTP map \blk $\Phi_{t,t'}$.   If, on the other hand, we are interested in characterizing Markovianity in terms of the failure of dynamical decoupling (DD), one has to (i) calculate or characterise the dynamical maps $\phi^{(DD)}_t (\cdot)$  resulting from evolving under $H_{SB}(t) + H^{(DD)}_S (t)$, where $H^{(DD)}_S (t)$ is the control Hamiltonian implementing a dynamical decoupling sequence, and (ii) verify that  $\phi^{(DD)}_t (\cdot) = \phi_t (\cdot)$ for all DD sequences. 
For the purposes of the upcoming discussion, let us label the set of routines that generate the required evolutions by $\mathcal{EM}$, and the classical processing algorithm that verifies ${\cal M}[H(t), \rho_B]$  by $\mathcal{CM}$. 

Let us now consider the dynamics of an initially non-factorisable state. If a system is computationally Markovian according to Definition~\ref{defdm}, relative to some ${\cal M}[H(t), \rho]$, \blk then there must exist a B+ decomposition of the form $\rho_{SB} = \sum_{a} \tilde{w}_a \blk  \tilde{Q}_a \otimes \tilde{\rho}_a$,  corresponding to some basis of positive operators $\{\tilde P_\alpha\}$, such that each pair $\{H(t), \tilde{\rho}_\alpha \}$ satisfies ${\cal M}[H(t), \tilde\rho_\alpha]$ for  $\tilde w_\alpha\neq 0$. Notice that the $\tilde{Q}_\alpha$  need not result from a canonical decomposition and thus there can be more than $d^2$ of them. Now, imagine that the experimenter has a preferred \blk $B+$ composition  corresponding to some \blk basis $\{ P_\alpha\}$. The associated  B+ decomposition, 
\begin{align*}
\rho_{SB} = \sum_{\alpha} w_\alpha Q_\alpha \otimes {\rho}_\alpha,
\end{align*}
induces the (experimentally accessible, as discussed earlier) bath states $\rho_\alpha$ and \blk CP maps $\phi^{(\alpha)}_t$. Thus, for each routine in $\mathcal{EM}$, which would give us the evolution of $\bullet\otimes\rho_B$  and associated set of CPTP maps $\{\phi_{t;i}\}$ \blk in the factorisable case, \blk we can now obtain the evolution of $\bullet\otimes\rho_\alpha$ for each $\alpha$  and the associated maps $\{\phi^{(\alpha)}_{t;i}\}$. \blk  So, in principle, one can verify if the preferred decomposition is in fact a decomposition in which the dynamics is computationally Markovian, by applying $\mathcal{CM}$ to each ${\cal M}[H(t), \rho_\alpha]$. 

One can take this a step further and search for the B+ decomposition in which the dynamics is computationally Markovian, by additional classical processing. Imagine one has already run $\mathcal{EM}$ with the preferred basis and has thus access to the evolution of $\bullet\otimes\rho_\alpha$ \blk for every $\alpha$, and thus all the maps $\{\phi^{(\alpha)}_{t;i}\}$ sufficient to verify$ \mathcal{M}[H(t), \rho_\alpha]$ for all $\alpha$.  \blk It follows that \blk the two decompositions, one induced by the decomposition in which the evolution is computationally Markovian and the other by the preferred decomposition, are related via $ \tilde{w}_a \tilde \rho_a = \sum_{\alpha} \tr{\tilde{P}_{a} Q_\alpha} w_\alpha \rho_\alpha$ and thus the corresponding sets of maps satisfy 
$$
\tilde{w}_a \tilde{\phi}^{(a)}_{t;i} = \sum_{\alpha} \tr{\tilde{P}_{a} Q_\alpha} w_\alpha \phi^{(\alpha)}_{t;i}.
$$

 This suggests then that one can search for a set of positive operators $\{{P'}_{a}\}$, or equivalently a set of coefficients $\kappa_{a,\alpha}=\tr{{P'}_{a} Q_\alpha}$, such that the maps $${w'}_a {\phi'}^{(a)}_{t;i} = \sum_{\alpha} \kappa_{a,\alpha} w_\alpha \phi^{(\alpha)}_{t;i}$$ are such that: (i) $\{{\phi'}^{(a)}_{t;i}\}$ satisfy the mathematical constraints associated to $\mathcal{M}[U_t, \rho'_a]$, where $\rho'_a = {\rm Tr_S} [ (P_a \otimes I_B)  \rho_{SB}]$, for all $a$, and (ii) that they lead to consistent evolutions, i.e., $\sum_{a} {w'}_a {\phi'}^{(a)}_{t;i} (\rho_S)= \sum_{a} {w}_a {\phi}^{(a)}_{t;i} (\rho_S)$. If one can find such set  $\{{P'}_{a}\}$ or the set of coefficients $\kappa_{a,\alpha}$, or even show that it exists, one says that the evolution is computationally Markovian (in particular with respect to the B+ decomposition associated to $\{{P'}_{a}\} = \{\tilde{P}_{a}\}$).  Thus, computational Markovianity can in principle be experimentally verified.

We stress that the search for the $\{{P'}_{a}\}$ is purely numerical, and that one only has to run $\mathcal{EM}$ once in order to get the $w_\alpha \phi^{(\alpha)}_{t;i}$, i.e., the additional cost of verifying computational Markovianity is classical processing. We expect the search to be a hard problem in general but, as we point out in the main text, it can also be trivial, depending on $\mathcal{M}[H(t),\rho_B]$.

\blk

\end{document}